\documentclass[journal]{IEEEtran}
\usepackage{cite}
  \usepackage{graphicx}
\graphicspath{{figs/}}  

\usepackage{amsmath}
\usepackage{amsthm}
\usepackage{amssymb}
\usepackage{amsfonts}
\usepackage{xcolor}
\usepackage{enumitem}
\usepackage{soul}

\usepackage{booktabs}
\newcommand*{\red}{\textcolor{red}}

 \definecolor{myOrange}{rgb}{1,0.5,0}

\usepackage{etoolbox}

\newtheorem{proposition}{Proposition}
\newtheorem{corollary}{Corollary}
\newtheorem{definition}{Definition}
\newtheorem{remark}{Remark}
\newtheorem{lemma}{Lemma}
\newtheorem{assumption}{Assumption}

\newtheorem{example}{Example}

\newtheorem{theorem}{Theorem}

\begin{document}
%
\title{Single module identifiability in linear dynamic networks with partial excitation and measurement}
%
%
%

\author{Shengling~Shi,~\IEEEmembership{Member,~IEEE,}
        Xiaodong~Cheng,~\IEEEmembership{Member,~IEEE,}
        and Paul~M.J.~Van~den~Hof,~\IEEEmembership{Fellow,~IEEE}
\thanks{\copyright 2021 IEEE. Personal use of this material is permitted.  Permission from IEEE must be obtained for all other uses, in any current or future media, including reprinting/republishing this material for advertising or promotional purposes, creating new collective works, for resale or redistribution to servers or lists, or reuse of any copyrighted component of this work in other works.  \newline Submitted to IEEE Trans. Automatic Control, 21 December 2020; revised 11 August 2021 and 19 October 2021. Accepted 18 December 2021. \newline   This work is supported by the European Research Council (ERC), Advanced Research Grant SYSDYNET, under the European Unions Horizon 2020 research and innovation programme (Grant Agreement No. 694504).}
\thanks{S. Shi and P. M. J. Van den Hof are with the Control Systems Group, Department of Electrical Engineering, 	Eindhoven University of Technology,
			5600 MB Eindhoven,
			The Netherlands. {\tt\small \{s.shi, p.m.j.vandenhof\}@tue.nl}}.
\thanks{X. Cheng was with the Control Systems Group, Department of Electrical Engineering, Eindhoven University of Technology, 5600 MB Eindhoven,
			The Netherlands. He is currently a research associate with the Control Group, Department of Engineering,
University of Cambridge, CB2 1PZ, United Kingdom. {\tt\small\{xc336@cam.ac.uk\} }}
}

\maketitle

\begin{abstract}
Identifiability of a single module in a network of transfer functions is determined by whether a particular transfer function in the network can be uniquely distinguished within a network model set, on the basis of data. Whereas previous research has focused on the situations that all network signals are either excited or measured, we develop generalized analysis results for the situation of partial measurement and partial excitation. As identifiability conditions typically require a sufficient number of external excitation signals, this work introduces a novel network model structure such that excitation from unmeasured noise signals is included, which leads to less conservative identifiability conditions than relying on measured excitation signals only. More importantly, graphical conditions are developed to verify global and generic identifiability of a single module based on the topology of the dynamic network. Depending on whether the input or the output of the module can be measured, we present four identifiability conditions which cover all possible situations in single module identification. These conditions further lead to synthesis approaches for allocating excitation signals and selecting measured signals, to warrant single module identifiability. In addition, if the identifiability conditions are satisfied for a sufficient number of external excitation signals only, indirect identification methods are developed to provide a consistent estimate of the module. All the obtained results are also extended to identifiability of multiple modules in the network.
\end{abstract}


%
\IEEEpeerreviewmaketitle

\section{Introduction}
%
%
%
%
\IEEEPARstart{D}{ue} to the increasing complexity of current technological systems, the study of large-scale interconnected dynamic systems receives considerable attention. They can adequately describe a wide class of complex systems in various applications, including multi-robot coordination \cite{mesbahi2010graph}, power grids \cite{XiaodongACOM2018Power}, gene networks \cite{adebayo2012dynamical} and brain networks \cite{Esch2020}. For data-driven modeling problems in structured dynamic networks, different types of network models have been used. Connecting to  prediction-error identification methods, the most popular modeling framework is based on a network of transfer functions as introduced in \cite{Goncalves&Warnick:08,van2013identification}, where vertices represent the internal signals, that can be measured, and directed edges denote transfer functions referred to as modules that represent the causal relations among the signals. This modeling framework has been applied to brain networks \cite{Esch2020}, gas pipeline networks \cite{dankers2021} and physical networks with diffusive couplings \cite{kivits2019dynamic}. While there are alternatives, e.g. in the form of state-space models \cite{haber2014subspace,Yu2018Subspace2,YU2019,henk2019topology}, this paper will adhere to the former so-called {\it module} framework. 

Various identification problems of dynamic networks have been addressed in the literature. One can e.g. focus on the estimation of the network topology, i.e. the interconnection structure of the network \cite{materassi2012problem,sanandaji2011acc,chiuso2012bayesian, hayden2016sparse,ShiTopology,Zorzi&Chiuso:17}. Another problem is to identify a single module in the network while the topology of the network is given. This includes the selection of internal signals that need to be measured and excited for achieving consistent module estimates \cite{van2013identification,dankers2015,gevers2015identification,linder2017identification,Everitt&Bottegal&Hjalmarsson_Autom:18,gevers2018practical,materassi2019signal,ramaswamy2019local}. Identification of the full network dynamics, for a given network topology, is addressed in, e.g., \cite{chiuso2012bayesian,weerts2018prediction}. 

In this paper, we focus on {\it network identifiability}, which is a concept that is independent of the particular identification method chosen. Based on the deterministic network reconstruction problems in \cite{Goncalves&Warnick:08,adebayo2012dynamical}, a novel concept of global network identifiability is introduced in an identification setting in \cite{Weerts&etal_IFAC:15, HarmAutomatica}, as a property that reflects the ability to distinguish between network models in a network model set on the basis of measurement data. In the literature, there are three notions of network identifiability which are introduced here from the strongest to the weakest: \textit{global identifiability}  \cite{Weerts&etal_IFAC:15, HarmAutomatica, henk2018necessary} that requires \textit{all} models to be distinguishable from all other models in the model set\footnote{There are actually two versions of global identifiability, reflecting whether either one particular model in the set can be distinguished or all models in the set \cite{HarmAutomatica}.};  \textit{generic identifiability} \cite{bazanella2017, hendrickx2018identifiability,chengallocation, 2019arXivChen, shi2020generic} which requires \textit{almost all} models to be distinguishable from all other models in the model set; and \textit{local identifiability} \cite{legat2020local} which requires models to be distinguishable from the models in a small neighborhood. In this work, we consider global identifiability and generic identifiability. 

In addition, identifiability of a network can be analyzed in a localized fashion by considering identifiability of each multiple-input-single-output (MISO) subsystem or each single-input-multiple-output (SIMO) subsystem independently \cite{HarmAutomatica,HarmGraphCDC,henk2018necessary, hendrickx2018identifiability}. On the other hand, identifiability of different MISO subsystems or SIMO subsystems can also be analyzed jointly based on the topology of the complete network \cite{chengallocation,2019arXivChen,bazanella2019network}. In this work, we consider the localized analysis of identifiability. Particularly, identifiability of a single module is investigated, and the results in this work are also extended to multiple modules from the same MISO or SIMO subsystem.

The above problem of global and generic identifiability has been investigated for different settings. In \cite{bazanella2017, hendrickx2018identifiability, henk2018necessary}, all vertices are excited by external excitation signals, while only a subset of vertices is measured. In contrast, the analysis in e.g., \cite{HarmAutomatica, shi2020generic} assumes that all vertices are measured, while a subset of vertices is excited, i.e. the so-called full measurement setting. However, the above works do not address the problem of global and generic identifiability when not all vertices are excited and not all vertices are measured, i.e. the setting with partial measurement and partial excitation. The analysis of generic identifiability in this setting has rarely been addressed. A recent work  \cite{cheng2021necessary} only addresses a necessary condition for network identifiability. In \cite{bazanella2019network}, sufficient conditions on generic identifiability are presented, requiring a sufficient number of excitation signals to achieve generic identifiability, where only measured excitation signals are considered. However, the contribution of the unmeasured noise signals, which may also serve as excitation sources for identifiability as shown in \cite{HarmAutomatica, shi2020generic} for the full measurement setting, is not explored. In addition, the main result in \cite{bazanella2019network} is not fully graphical as it also requires the availability of certain mappings from excitation signals to node signals, and thus the conditions cannot be tested solely based on the network topology. Special network structures are then considered in \cite{bazanella2019network} such that the required mappings are obtainable; however, how to handle networks with more general topology is not addressed. 

In this paper, we consider the network identifiability concept in \cite{HarmAutomatica, shi2020IFAC, shi2020generic} and generalize the results significantly from the full measurement setting to the setting with partial measurement and partial excitation. We also address the limitations of \cite{bazanella2019network} by exploring the excitation contributed by unmeasured noises and by developing fully graphical identifiability conditions. Additionally, the model sets considered are allowed to contain a priori known/fixed modules.

In Section~\ref{sec:Noise} we show how unmeasured noise signals can serve as excitation sources for identifiability analysis. This analysis has been performed in \cite{HarmAutomatica,shi2020generic} with all the vertices measured, but the results are not directly applicable to networks with unmeasured vertices. In this work, these results are generalized to the partial measurement case by introducing a concept of equivalence between network models and by developing a novel network model structure. Due to the contribution of the noise signals, a smaller number of measured excitation signals is needed for network identifiability, compared to the result in \cite{bazanella2019network}.
 
More importantly, for the developed model structure, this work develops a series of novel graphical conditions to analyze both global and generic identifiability of a single module with different excitation and measurement schemes in Sections~\ref{sec:nece}, \ref{sec:suffi} and \ref{sec:Un-inputorUn-output}. With the obtained conditions, single module identifiability can be checked by only inspecting the topology of the dynamic network. It is worth emphasizing that the conditions presented in this paper cover all possible measurement schemes of the input and the output when identifying the target module in a network. In addition, the graphical conditions further lead to comprehensive synthesis approaches in Section~\ref{sec:synthesis} for excitation and sensor allocation to achieve identifiability, and indirect identification methods for single module estimation in Section~\ref{sec:Indirec}. All the above results also extend to multiple modules from the same MISO or SIMO subsystem of the network. The proofs of all the technical results are collected in the appendix.

\section{Preliminaries and problem formulation} \label{sec:preAndPro}
\subsection{Dynamic networks}
The dynamic network model describes the relationship among a set of $L$ scalar \textit{internal signals} $ \mathcal{W} \triangleq \{w_1(t),w_2(t),\cdots,w_L(t)\}$, a set of deterministic excitation signals $ \{r_1(t),\cdots,r_K(t)\}$ with $K \leqslant L$, and a set of unmeasured disturbances $\{v_1(t),\cdots,v_L(t)\}$. The model is written as 
\begin{align}
w(t) &= G(q)w(t)+R r(t)+v(t), \nonumber \\
w_{\mathcal{C}}(t) &= C w(t), \label{eq:origiNet}
\end{align}
where $G(q)$ is a $L \times L$ matrix of rational transfer operators with delay operator $q^{-1}$, i.e.   $q^{-1} w_i(t) = w_i(t-1)$; $w(t)$, $r(t)$ and $v(t)$ are the column vectors that collect all the internal signals, excitation signals and disturbances, respectively; $C$ is a binary matrix that extracts a subvector $w_\mathcal{C}(t)$ from $w(t)$, i.e. $C$ consists of a subset of rows of an $L \times L$ identity matrix; $R$ is a binary matrix that decides which internal signals are influenced by $r(t)$, i.e. $R$ consists of $K$ columns of an $L \times L$ identity matrix. In the above model, only $w_\mathcal{C}(t)$ and $r(t)$ signals are measured.

In addition, $v(t)$ is a vector of zero-mean stationary stochastic processes. Let $\Phi_v(q)$ of dimension $L \times L$ denote the rational power spectral density matrix of $v(t)$ with rank $p \leqslant L$, and then a noise model for $v(t)$ can be introduced based on the spectral factorization of $\Phi_v(q)$ as
\begin{equation}\label{eq:noisemodel}
v(t)=H(q) e(t),
\end{equation}
where $ e(t)$ is a vector of white noises with a covariance matrix $\Lambda$; $H(q)$ is proper and stable. Depending on the chosen spectral factorization of $\Phi_v$, $\Lambda$ may have size $L$ or $p$ \cite{HarmAutomatica, Gevers2019Singular}.

Combining \eqref{eq:origiNet} and \eqref{eq:noisemodel} leads to a complete network model specified as a quintuple $M \triangleq (G(q),R,C,H(q),\Lambda)$, on which the following assumptions are made:
\begin{assumption}
\label{ass1}
It will be assumed that
\begin{enumerate}[label=(\alph*)]
\item $G(q)$ is proper, stable, and hollow, i.e., the entries on its main diagonal are zeros;

\item $[I-G(q)]^{-1}$ is stable;

\item The network is well-posed in the sense that all principal minors of $\lim_{z \to \infty}(I-G(z))$ are non-zero \cite{dankers2014system};

\item $H(q)$ is proper and stable;

\item $\Lambda$ is real and positive semi-definite.
\end{enumerate}
\end{assumption}
Note that Assumption~\ref{ass1}(c) ensures that every principal submatrix of $[I-G(q)]$ has a proper inverse \cite{scherer2001theory}, i.e. every 
 closed-loop transfer function is proper.

In a network model, both the excitation signals in $r(t)$ and the noise signals in $e(t)$ are called \textit{external signals}, which are collected in a set $\mathcal{X}$. The entries in $G(q)$ are referred to as \textit{modules}. Let set $\mathcal{C} \subseteq \mathcal{W}$ contain all the measured internal signals in $w_\mathcal{C}(t)$, and $\mathcal{Z} = \mathcal{W} \setminus \mathcal{C}$ contains all the unmeasured internal signals. Without loss of generality, vector $w(t)$ is considered to be ordered as $w(t) = \begin{bmatrix}
w_\mathcal{C}(t) \\
w_\mathcal{Z}(t)
\end{bmatrix}$, where $w_\mathcal{C}(t)$ and $w_\mathcal{Z}(t)$ contain the measured internal signals and the unmeasured ones. Accordingly, $C$ is partitioned as $C = \begin{bmatrix}
I & 0
\end{bmatrix}.$

The external-to-internal mapping of \eqref{eq:origiNet} is
\begin{equation}
w_{\mathcal{C}} = C[I-G(q)]^{-1}Rr + C[I-G(q)]^{-1}H(q)e,  \label{eq:origiNetMap}
\end{equation} 
and a standard open-loop identification of the above model \cite{ljung1999system} can typically lead to consistent estimates of the following objects:
\begin{equation}
C T(q) R, \quad C\Phi(\omega) C^T, \label{eq:objects}
\end{equation}
where $T(q) \triangleq [I-G(q)]^{-1}$, $
\Phi(\omega)  \triangleq [I-G(e^ {i\omega})]^{-1}H(e^ {i\omega}) \Lambda H(e^{i\omega})^\star[I-G(e^ {i\omega})]^{-\star}$, and $(\cdot)^\star$ denotes the complex conjugate transpose. Note that $CT(q)R$ contains a subset of rows and columns of $T(q)$, based on which internal signals are measured or excited. In addition, it can be found that the two objects in \eqref{eq:objects} describe the stochastic properties of the measured internal signals $w_\mathcal{C}$, i.e. its mean and spectral density \cite{HarmAutomatica}. Thus, an identifiability question arises to determine the uniqueness of modules in $G(q)$ given the objects in \eqref{eq:objects}. In addition, we consider the situation where the measured internal signals are affected by a full-rank process noise.

\begin{assumption} \label{ass:fullrankPhi}
The power spectrum $C\Phi(\omega) C^T$ has full rank for almost all $\omega$.
\end{assumption}

\subsection{Model sets and identifiability} \label{sec:modelSet}
Network identifiability is defined on the basis of a network model set whose definition is given first. For a network model $M$ and by parameterizing its entries in a rational form as $M(\theta) =  (G(q,\theta),R,C, H(q,\theta),\Lambda(\theta))$, a network model set $\mathcal{M} \triangleq  \{M(\theta)|\theta \in \Theta \subseteq \mathbb{R}^n \}$ is formulated, where $M(\theta)$ satisfies Assumption~\ref{ass1} for every $\theta \in \Theta$. \textit{Note that in the sequel the dependency of transfer matrices on $q$ and $\theta$ is often omitted for simplicity of notation when no confusion arises.} For the graphical identifiability analysis, the transfer functions in $G$ and $H$ are typically parameterized with independent parameters \cite{HarmAutomatica}.
\begin{assumption}\label{ass:InPara}
The transfer functions in a parameterized network model set are parameterized independently.
\end{assumption}
%

There can be certain entries in the matrices of $M(\theta)$ that are fixed and thus do not depend on the parameters. These entries are called \textit{known} or {\it fixed} modules, which reflect the prior knowledge or simply the modeling assumptions imposed by the user. For example, the absence of an interconnection between internal signals is represented by a fixed $0$ in $G$; some entries in $G$ may be particularly designed controllers that are \textit{known}, while in $H(q)$ entries can be $1$ or $0$ specifying where $e$ signals enter the network. The entries that depend on the parameters are called \textit{unknown} or {\it parameterized} entries.

The above structural information of a model set can be represented by a directed graph $\mathcal{G} = (\mathcal{V},\mathcal{E})$, where $\mathcal{V} \triangleq \mathcal{W} \cup  \mathcal{X}$ is a set of vertices representing both the internal signals in $w$ and the external signals in $r,e$, and $\mathcal{E} \subseteq \mathcal{V} \times \mathcal{V}$ denotes a set of directed edges representing those entries in $(G,R,H)$ that are not fixed to zero, e.g. the directed edge from $w_i$ to $w_j$, denoted by $(w_i,w_j)$, exists if and only if $G_{ji}$ is not fixed to zero, and this edge is said to be known (or unknown) if $G_{ji}$ is known (or unknown). Similarly, $(e_k,w_j) \in \mathcal{E}$ and $(r_p,w_j) \in \mathcal{E}$ if and only if $H_{jk}$ and $R_{jp}$ are not fixed to zero, respectively. In this way, any parameterized model set or network model induces a directed graph $\mathcal{G}$ to encode its structural information. Note that notation $w_i$ now represents both a signal and a vertex, and the dependency of the signal on $t$ is often omitted for the simplicity of notation.

Concerning network identifiability, we follow the concept of global network identifiability as defined in \cite{HarmAutomatica} and also consider its generic version obtained by combining it with the concept of generic identifiability that was originally introduced in \cite{bazanella2017, hendrickx2018identifiability} for a different setting. In this respect, we follow an approach that is formulated in \cite{shi2020generic}.
\begin{definition}
 \label{def:defnOriIden}
Given a parameterized network model set $\mathcal{M}$, consider $\theta_0 \in \Theta$ and the following implication:
\begin{equation} \label{equivTP}
		\left. \begin{array}{c} CT(q,\theta_0)R= CT(q,\theta_1)R \\ C\Phi(\omega,\theta_0)C^\top= C\Phi(\omega,\theta_1)C^\top \end{array} \right\}
		\Rightarrow
 G_{ji}(q,\theta_0) = G_{ji}(q,\theta_1),
\end{equation}
for all $\theta_1 \in\Theta$. Then the module $G_{ji}(q,\theta)$ is
\begin{itemize}
\item globally identifiable in $\mathcal{M}$ from $(w_\mathcal{C},r)$ if the implication (\ref{equivTP}) holds for all $\theta_0 \in \Theta$;
\item generically identifiable in $\mathcal{M}$ from $(w_\mathcal{C},r)$ if the implication (\ref{equivTP}) holds for almost all $\theta_0 \in \Theta$.
\end{itemize}
\end{definition}
In the above definition, the notion ``almost all'' excludes a subset of Lebesgue measure zero from $\Theta$. The definition based on a single module also extends trivially to a subset of modules, i.e. a subnetwork, by replacing $G_{ji}$ in the RHS of \eqref{equivTP}. The concept of identifiability in this definition concerns the uniqueness of a module given the first and second moment information of the measured signals \cite{HarmAutomatica}. If the module is not identifiable in the model set, no identification method, that is based on the first and second moments for estimating the module, can guarantee to provide a unique estimate of the module. This identifiability issue is illustrated in the following example.
\begin{example} \label{example0}
Consider the graph of a model set in Fig.~\ref{fig:exam0}(a), where identifiability of $G_{31}$ is of interest with measured input and output. The question is whether $G_{31}$ can be uniquely distinguished within the model set, on the basis of $CTR$ and $C\Phi C^\top$. As there is no noise, we consider $CTR$ only, i.e. the mappings from $r_1$ to the measured internal signals $w_1$ and $w_3$, denoted by $T_{11}$ and $T_{31}$, respectively.

\begin{figure}[h]
\begin{minipage}{0.23\textwidth}
\centering
\includegraphics[scale=0.34]{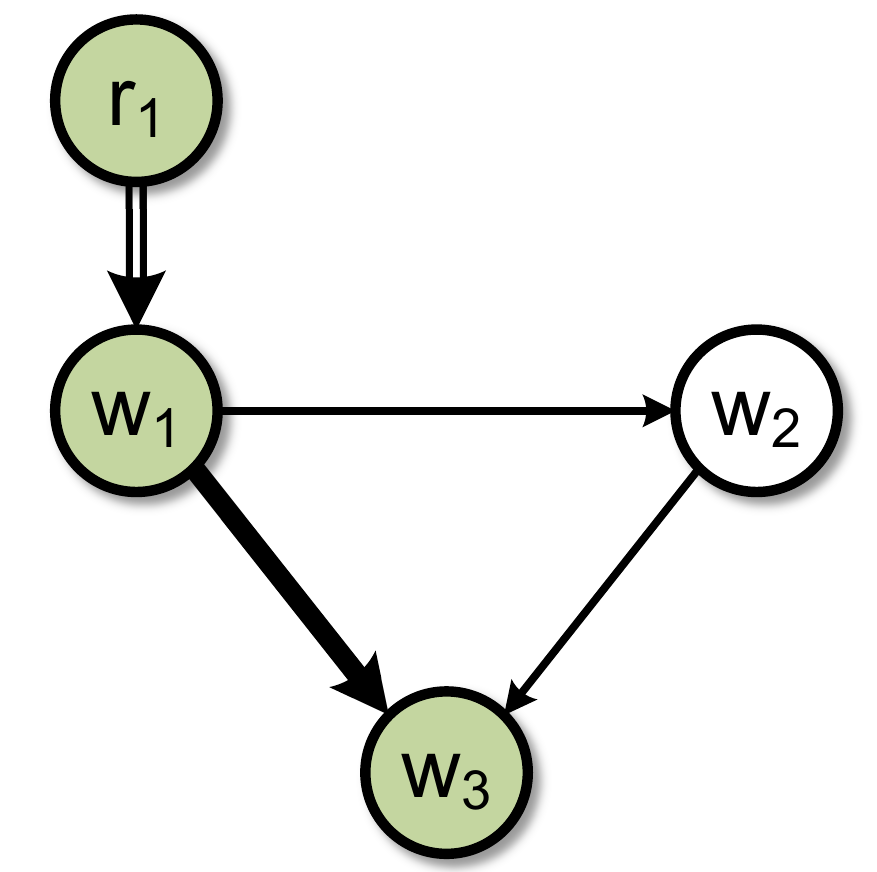}
\\(a)
\end{minipage}
\begin{minipage}{0.23\textwidth}
\centering
\includegraphics[scale=0.34]{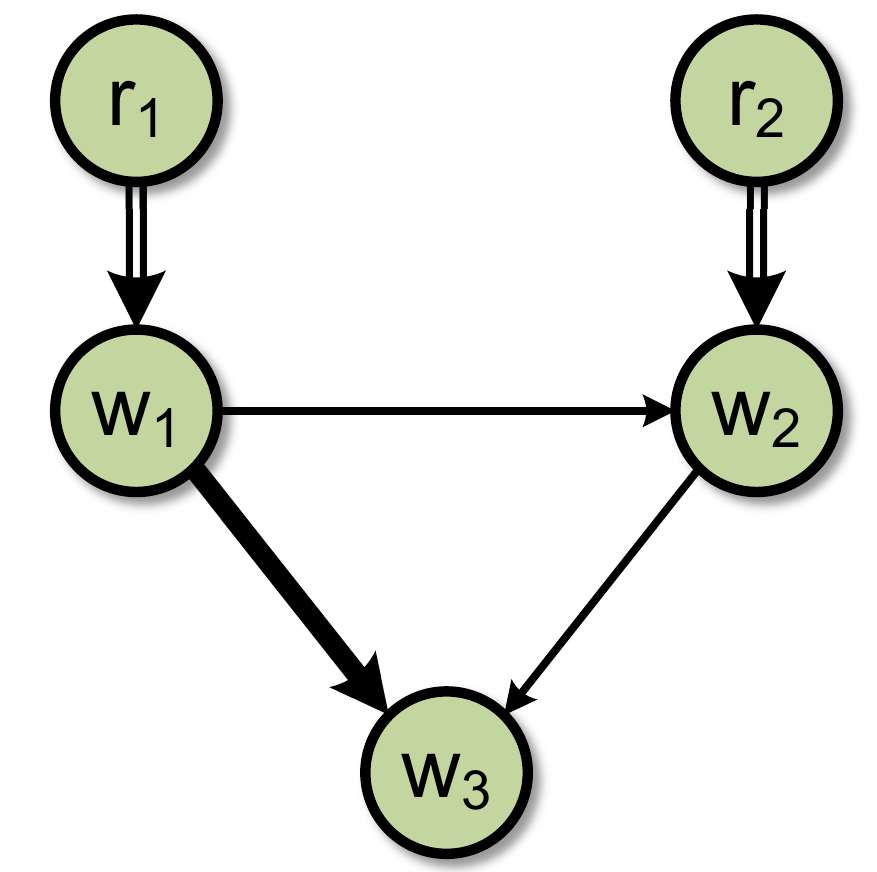}
\\(b)
\end{minipage}
\caption{Two network model sets with $G_{31}$ as the target module (thick edge) and known $R_{11}$, $R_{22}$ (double-line edges). The input and the output of $G_{31}$ are measured, indicated in green; however, $w_2$ is unmeasured in (a). }
\label{fig:exam0}
\end{figure}
Since $R$ is a binary matrix, it is known that $T_{11} = R_{11}=1$, and furthermore, it holds that
\begin{equation}
T_{31} = G_{31}+ G_{32}G_{21}. \label{eq:example0}
\end{equation}
Since both $G_{32}$ and $G_{21}$ are unknown, the module $G_{31}$ cannot be uniquely recovered from $T_{11}$ and $T_{31}$ and thus is not identifiable.

However, when $w_2$ is also measured and excited as in Fig.~\ref{fig:exam0}(b), the additional mappings given by $CTR$, i.e.
$$
T_{21} = G_{21}, \quad T_{32} = G_{32},
$$
together with $T_{31}$ in \eqref{eq:example0} ensure the uniqueness of $G_{31}$ as
$$
G_{31} = T_{31} - T_{32} T_{21}.
$$
Therefore, $G_{31}$ is globally identifiable in this model set.
\end{example}

\subsection{Problem formulation}
In this paper, we are going to investigate under which conditions a module $G_{ji}$ is identifiable in a model set $\mathcal{M}$, on the basis of measured signals $w_\mathcal{C}$ and $r$. While focusing on generic identifiability, we will develop conditions that are fully graph-based and that can handle flexible model sets including prior known/fixed modules.

Since identifiability conditions typically require sufficient excitation signals \cite{HarmAutomatica,shi2020generic}, a challenge will be to explore how noise signals can be utilized for this purpose. This will require a further analysis of the spectral factorization of the noise power spectrum $C \Phi C^\top$ in \eqref{equivTP}. 

\subsection{Notations and definitions}
\label{sec:notation}
The following notations are used throughout the paper. Matrix $T(q,\theta)$ is called to have full rank generically (or full rank globally) if it has full rank for almost all $\theta$ (or for all $\theta$). More generally, a property that depends on parameter $\theta$ is said to hold generically (or globally) if it holds for almost all $\theta$ (or for all $\theta$). For subsets $\mathcal{W}_1$, $\mathcal{W}_2 \subseteq \mathcal{W}$, $T_{\mathcal{W}_1 \mathcal{W}_2}$ denotes the mapping from the internal signals in $\mathcal{W}_2$ to the ones in $\mathcal{W}_1$, i.e. a submatrix of $T$ in \eqref{eq:objects} whose rows and columns correspond to the signals in $\mathcal{W}_1$ and $\mathcal{W}_2$, respectively. A set in the subscript with a single element is replaced by the index of this element, i.e. $T_{\mathcal{W}_1 \{w_i\}}$ is simply written as $T_{\mathcal{W}_1 i}$. A binary matrix is a selection matrix if it consists of a subset of rows from an identity matrix. By pre-multiplying a matrix $A$ by the selection matrix, a subset of rows in $A$ can be extracted.

In a graph $\mathcal{G}$, a directed edge from $w_i$ to $w_j$, i.e. $(w_i,w_j)$, is called an \textit{in-coming} edge of $w_j$, and an \textit{out-going} edge of $w_i$. In this case, $w_i$ is called an in-neighbor of $w_j$, and $w_j$ is an out-neighbor of $w_i$. The out-degree of $w_i$ is the total number of out-neighbors of $w_i$. A (directed) \textit{path} from $w_i$ to $w_j$ is a sequence of vertices and out-going edges starting from $w_i$ to $w_j$ without repeating any vertex.  The \textit{length} of a path is the number of edges in the path. A single vertex is also regarded as a directed path to itself. In a path, \textit{internal vertices} are the vertices excluding the starting and the ending vertices. A vertex $w_i$ is said to be excited by a vertex $x_i \in \mathcal{X}$ if there is a directed edge $(x_i,w_i)$, and $w_i$ is indirectly excited by $x_i$ if there exists a path from $x_i$ to $w_i$ with length larger than one. Similarly, $w_i$ is said to be measured if $w_i \in \mathcal{C}$, and it is indirectly measured if $w_i \notin \mathcal{C}$ but it has a path to a measured internal signal.

Two directed paths are called \textit{vertex disjoint} if they do not share any vertex, including the starting and ending vertices, otherwise they \textit{intersect}. Given two subsets of vertices $\mathcal{V}_1$ and $\mathcal{V}_2$, $b_{\mathcal{V}_1 \to \mathcal{V}_2}$ denotes the maximum number of vertex disjoint paths from $\mathcal{V}_1$ to $\mathcal{V}_2$. A vertex set $\mathcal{D}$ is a $\mathcal{V}_1 - \mathcal{V}_2$ \textit{disconnecting set} if it intersects with all paths from $\mathcal{V}_1$ to $\mathcal{V}_2$, where $\mathcal{D}$ may also include vertices in $\mathcal{V}_1 \cup \mathcal{V}_2$. It is a minimum disconnecting set if it has the minimum cardinality among all $\mathcal{V}_1-\mathcal{V}_2$ disconnecting sets \cite{schrijver2003combinatorial}.

Important set notations in this work are collected in Table~\ref{table:Notation}, and some of them will also be formally introduced later in the main results.

\begin{table}[h]
\caption{Description of important set notations.} \label{table:Notation}
\begin{tabular}{ll}
\toprule
$\mathcal{X}$             & All the external signals.\\
$\mathcal{X}_j$             & All the external signals that do not have an unknown edge to $w_j$. \\
$\mathcal{W}$             & All the internal signals. \\
$\mathcal{C}$             & All the measured internal signals. \\
$\mathcal{W}^-_j$        & The internal signals that have an unknown directed edge to $w_j$.\\
$\mathcal{W}^+_i$           & The internal signals to which $w_i$ has an unknown directed edge.\\
$\mathcal{N}_j^-$        &  All internal signals that are in-neighbors of $w_j$ excluding \\ & the measured signals that have a known directed edge to $w_j$. \\
$\mathcal{N}_i^+$            & The out-neighbors of $w_i$ excluding the ones that satisfy \\ & (i) $w_i$ has a known directed edge to it, \textit{and}  (ii) it is excited \\ & by an external signal that has out-degree one  \\ & with a known out-going edge.  \\
\bottomrule
\end{tabular}
\end{table}

\section{Equivalent network for noise excitation} \label{sec:Noise}
Identifiability in Definition~\ref{def:defnOriIden} concerns the uniqueness of a single module given the mapping $CTR$ and the spectrum $C \Phi C^\top$. In this section, we introduce a novel network model structure to model the data without loss of generality. With this model structure, the noise spectrum $C \Phi C^\top$ admits a unique spectral factor $CT \tilde{H}$ for a transformed noise model $\tilde{H}$. Therefore, implication~\eqref{equivTP} can be equivalently simplified by considering $(CTR,CT \tilde{H})$ in the LHS instead of $(CTR, C \Phi C^\top)$, which implies that the unmeasured noises can act as excitation sources for the identifiability analysis. 

\subsection{Noise spectrum analysis and equivalent networks}
We introduce the novel model structure by exploiting a concept of network equivalence. Based on \eqref{eq:origiNetMap}, the objects $CTRr$ and $C \Phi C^\top$ reflect the mean and the power spectral density of the measured process $w_\mathcal{C}$. These objects encode all the stochastic properties of interest for the measured processes $(w_\mathcal{C},r)$, as the first and the second moments are the main focus. This motivates the concept of network equivalence by extending \cite[Definition~4]{weerts2019abstractions} to the setting with partial measurement.

\begin{definition} \label{defn:equiv}
Any two network models $M_1=(G_1(q),R_1,C_1,\\ H_1(q), \Lambda_1)$ and $M_2=(G_2(q),R_2,C_2, H_2(q), \Lambda_2)$ are said to be (observationally) equivalent if it holds that 
\begin{align*}
C_1 T_1(q) R_1 = C_2 T_2(q) R_2, \text{ and } C_1\Phi_1(\omega) C_1^T = C_2 \Phi_2(\omega) C_2^T,
\end{align*}
where $T(q)$ and $\Phi(\omega)$ are defined in \eqref{eq:objects}, and the equivalence is denoted by $M_1 \sim M_2$.
\end{definition}

The above concept of equivalence characterizes two network models that can be used to model the same measured processes $(w_\mathcal{C},r)$, because given measured $r$, the stochastic processes $w_\mathcal{C}$ in two equivalent models have the same mean $CTRr$ and power spectrum $C \Phi C^\top$. Note that $G_1$ and $G_2$ from two equivalent models may have different dimensions, e.g. a model $M$ and its immersed network where $w_\mathcal{Z}$ is eliminated \cite{dankers2015, weerts2019abstractions}.

It can be found that any network admits the following equivalent network by exploiting the spectral factorization of the disturbance spectrum $C \Phi C^\top$.
\begin{theorem} \label{theorem:equiv}
For any network model $M = (G(q),R,C,H(q),\\ \Lambda)$, there exists an equivalent network model as \begin{equation}
\tilde{M} \triangleq (G(q),R,C,\begin{bmatrix} \tilde{H}^\star(q) & 0 \end{bmatrix}^\star,\tilde{\Lambda}), \label{eq:netNh}
\end{equation}
where $\tilde{H}(q)$ is a $c \times c$ rational transfer matrix, with $c = |\mathcal{C}|$, and is minimum phase and monic; $\tilde{\Lambda} \in \mathbb{R}^{c\times c}$ is positive semi-definite. In addition, 
\begin{itemize}
\item if $M$ satisfies Assumption~\ref{ass:fullrankPhi}, $(\tilde{H}(q),\tilde{\Lambda})$ is unique with positive definite $\tilde{\Lambda}$.
\end{itemize}
\end{theorem}

Based on the above result, the measured process $(w_\mathcal{C},r)$ that is modeled by $M$ can be equivalently modeled by $\tilde{M}$ in \eqref{eq:netNh}, which has the same matrices $G$, $R$, $C$; however, the unmeasured internal signals of $\tilde{M}$ are noise-free. Note that the model $\tilde{M}$ \eqref{eq:netNh} has a related white noise vector, denoted by $\tilde{e}$, with covariance matrix $\tilde{\Lambda}$. This noise model is simpler than the one in $M$, and more importantly, $\tilde{M}$ keeps the $G$ matrix invariant as in $M$. This invariance is important for identifiability analysis and identification of network modules.

The equivalence between $M$ and $\tilde{M}$ is obtained due to the freedom in transforming the unmeasured internal signals and modeling the noises, since the objects in \eqref{eq:objects} only reflect the properties of the measured processes.
\subsection{Equivalent network for handling noise excitation}
Since a network $M$ and its corresponding $\tilde{M}$ are equivalent and contain the same $G$ matrix, both of them can be used to model the same data set, i.e. the measured $(w_\mathcal{C},r)$, for the identification of the modules in a dynamic network \eqref{eq:origiNet}. In the previous section, it is discussed that $\tilde{M}$ in \eqref{eq:netNh} can potentially be a better option due to its simpler noise model. In this section, we further show that the particular noise model of $\tilde{M}$ is also beneficial for the identifiability analysis.

From now on, we use $\mathcal{M}$ to specifically refer to a parameterized set of models that are structured according to $\tilde{M}$ \eqref{eq:netNh}, defined as follows.

\begin{definition}
Let $\mathcal{M}$ be a network model set that is obtained from the rational parameterization of $\tilde{M}$ in \eqref{eq:netNh} as 
$$
\mathcal{M} \triangleq \{ \tilde{M}(\theta)| \theta \in \Theta \subseteq \mathbb{R}^n\}, 
$$
where $\tilde{M}(\theta)$ satisfies Assumption~\ref{ass1} for all $\theta \in \Theta$.
\end{definition}

It can be found that the implication~\eqref{equivTP} for $\mathcal{M}$ can be further simplified under mild conditions.

\begin{assumption} \label{ass:FeedThrough} 
A network model set $\mathcal{M}$ satisfies at least one of the following two conditions:
\begin{enumerate}[label=(\alph*)]
\item $G(q,\theta)$ is parameterized to be strictly proper;

\item $\tilde{\Lambda}$ is diagonal and $G(q,\theta)$ is parameterized without algebraic loops, i.e. there exists a permutation matrix $P$ such that $P G^\infty(\theta) P^\top$ is lower triangular, where $G^\infty(\theta) \triangleq \lim_{z \to \infty} G(z,\theta)$.
\end{enumerate}
\end{assumption}

\begin{proposition}
\label{pro:reforIden}
For a network model set $\mathcal{M}$  that satisfies Assumptions~\ref{ass:fullrankPhi}, \ref{ass:FeedThrough} and defining
\begin{equation}
T_{\mathcal{W}\mathcal{X}}(q,\theta) \triangleq  T(q,\theta)X(q,\theta), \quad X(q) \triangleq [R \quad \begin{bmatrix}
\tilde{H}(q,\theta) \\
0
\end{bmatrix}], \label{eq:TwxDef}
\end{equation}
implication \eqref{equivTP} for $\mathcal{M}$ can be equivalently formulated as
\begin{equation}
		CT_{\mathcal{W} \mathcal{X}}(q, \theta_0) = CT_{\mathcal{W} \mathcal{X}}(q,\theta_1) \Rightarrow
 G_{ji}(q,\theta_0) = G_{ji}(q,\theta_1), \label{eq:EquiIDconcept}
\end{equation}
for all $\theta_1 \in \Theta$.
\end{proposition}

According to \eqref{eq:origiNetMap}, the above result indicates that both the mappings from $r$ and $\tilde{e}$ to the measured internal signals can be used for analyzing identifiability in $\mathcal{M}$, and thus the unmeasured noise signal $\tilde{e}$ plays the same role as the measured $r(t)$ for the identifiability analysis. In this case, we say that $\tilde{e}$ signals act as excitation sources for the identifiability analysis. Proposition~\ref{pro:reforIden} is an extension of \cite[Propositions~2]{HarmAutomatica} to the partial measurement and partial excitation setting.

    Proposition~\ref{pro:reforIden} shows another advantage of $\tilde{M}$ over a general network model $M$ in network identification with partial measurement and partial excitation. These two models are equivalent to describe the same data and contain the same modules; however, the model set $\mathcal{M}$ of $\tilde{M}$ allows us to exploit the noise spectral density through Proposition~\ref{pro:reforIden}, such that the noise signals act as excitation signals for identifiability analysis. Therefore, in this work, we regard $\tilde{M}$ as the standard model for network identification in the partial measurement and partial excitation setting.

\section{Necessary graphical conditions} \label{sec:nece}
From now on, we focus on the development of graphical conditions for identifiability in $\mathcal{M}$ obtained from the parameterization of $\tilde{M}$. Particularly, we focus on identifiability of a single module.

Necessary and sufficient graphical conditions for generic identifiability of a single module are obtained in \cite{hendrickx2018identifiability} for the full excitation case and in \cite{shi2020generic} for the full measurement case. When the setting with partial measurement and partial excitation is considered, the existing necessary conditions from the above works for full measurement or full excitation setting naturally remain necessary conditions for the current setting. 

To introduce these necessary conditions, we first impose a regularity assumption on the fixed modules, since the fixed modules may cause graphical conditions to fail\cite{shi2020generic}. Firstly, we define that the structural rank of a matrix is the maximum rank of all matrices with the same nonzero pattern \cite{steffen2005control}.

\begin{assumption}[\cite{shi2020generic}]
\label{ass:fixEntry}
In model set $\mathcal{M}$, the rank of any submatrix of $[(G(q,\theta)-I) \text{ } X(q,\theta)]$ that does not depend on $\theta$ is equal to its structural rank.
\end{assumption}

Then define the following sets of signals related to the input $w_i$ and the output $w_j$ of $G_{ji}$:
\begin{itemize}
\item Let $\mathcal{X}_j$ contain all the signals in $r$ and $\tilde{e}$ that have \textit{no unknown} directed edge to $w_j$;

\item $\mathcal{W}^-_j$ contains all the internal signals that have an \textit{unknown} directed edge to $w_j$;

\item $\mathcal{W}^+_i$ contains all the internal signals to which $w_i$ has an \textit{unknown} directed edge.
\end{itemize}

Note that when all the non-zero modules are unknown, $\mathcal{W}^-_j$ and $\mathcal{W}^+_i$ contains all the in-neighbors of $w_j$ in $\mathcal{W}$ and all the out-neighbors of $w_i$, respectively.
Then the following necessary condition on excitation signals can be derived from \cite{hendrickx2018identifiability} and \cite{shi2020generic} directly.
\begin{lemma} \label{lemma:nece1}
In the model set $\mathcal{M}$ that satisfies Assumptions~\ref{ass:fullrankPhi}, \ref{ass:InPara}, \ref{ass:FeedThrough}, \ref{ass:fixEntry} and a topological property\footnote{For the necessity to hold, $\mathcal{M}$ should be an open set as in \cite[Assumption~4]{shi2020generic}.}, module $G_{ji}(q,\theta)$ is generically identifiable in $\mathcal{M}$ from $(w_\mathcal{C},r)$ only if the following conditions are satisfied:
\begin{enumerate}
\item \label{nece1path} $b_{\mathcal{X}_j \to \mathcal{W}^-_j} = 1+ b_{\mathcal{X}_j \to \mathcal{W}^-_j \setminus \{w_i\}} $;

\item \label{nece2path} $b_{\mathcal{W}^+_i \to \mathcal{C}} = 1+ b_{\mathcal{W}^+_i \setminus \{w_j\} \to \mathcal{C}}$;

\item Each signal in $\{w_i,w_j\}$ is measured or is excited by a vertex in $\mathcal{X}$.
\end{enumerate}
\end{lemma}

The above conditions imply a necessary number of measured signals and external signals, including $r$ and $\tilde{e}$, for identifiability. This is because the scheme of measurement and excitation decides the sparsity pattern of matrices $C$, $R$, $H$, and thus further influences the mapping $C T X$ in \eqref{eq:EquiIDconcept}. The formulation of $\mathcal{X}_j$ indicates that the noises, which have unknown directed edges to $w_j$, are not helpful for identifiability of $G_{ji}$.

The graphical conditions can be easily tested using graphical algorithms to compute the maximum number of vertex disjoint paths. However, the conditions are not suitable for designing synthesis approaches for excitation and sensor allocation, since they do not specify explicitly which signals are necessary to be excited and measured. Thus, following \cite{shi2020generic}, the above path-based conditions can be equivalently formulated with the concept of disconnecting sets. 

\begin{lemma} \label{lemma:nece2}
Consider the setting of Lemma~\ref{lemma:nece1},
\begin{enumerate}
\item \label{nece1Set} condition~\eqref{nece1path} holds if and only if there exists a $\mathcal{X}_j - \mathcal{W}^-_j \setminus \{w_i\}$ disconnecting set $\mathcal{D}$ such that 
$
b_{\mathcal{X}_j \to \{w_i\} \cup \mathcal{D}}=|\mathcal{D}|+1;
$

\item \label{nece2Set} condition~\eqref{nece2path} holds if and only if there exists a $\mathcal{W}^+_i \setminus \{w_j\} - \mathcal{C}$ disconnecting set $\mathcal{D}_c$ such that 
$
b_{\{w_j\} \cup \mathcal{D}_c \to  \mathcal{C}}=|\mathcal{D}_c|+1.
$
\end{enumerate}
\end{lemma}
\begin{proof}
The first result follows from \cite[Lemma~8]{shi2020generic} and condition~\eqref{nece1path} of Lemma~\ref{lemma:nece1}. The last result is the dual situation
\end{proof}

In the above lemma, the first result shows that the signals in $\{w_i\} \cup \mathcal{D}$ are necessary to be excited or indirectly excited by $r$ and $\tilde{e}$. In addition, the second result specifies that the signals in $\{w_j\} \cup \mathcal{D}_c$ should be measured or indirectly measured, i.e. they are not measured but have vertex disjoint paths to measured internal signals. 

However, the necessary conditions in Lemma~\ref{lemma:nece1} are not sufficient to verify identifiability. This also means that the requirements on excitation signals and on measured internal signals are not separable for identifiability, i.e. first allocating excitation signals according to the results for the full measurement case and then selecting measured signals according to the results for the full excitation case are not sufficient for identifiability in the current setting.

\section{Sufficient conditions: Both input and output measured or excited} \label{sec:suffi}
In this section, sufficient graphical conditions are developed to verify global and generic identifiability of a single module in $\mathcal{M}$ for the situation, where the input and the output are both measured or both excited. As shown in Proposition~\ref{pro:reforIden}, identifiability concerns the uniqueness of network modules given $CT_{\mathcal{W}\mathcal{X}}$. In the special cases where $C=I$ or $R=I$, identifiability of modules relates to the rank of submatrices in $T_{\mathcal{W}\mathcal{X}}$ as follows \cite{hendrickx2018identifiability,
shi2020generic}. Taking the case where $C=I$ as an example and based on \eqref{eq:EquiIDconcept} and the definition of $T_{\mathcal{W}\mathcal{X}}$ in \eqref{eq:TwxDef}, the rank condition is analyzed on the basis of the relation
\begin{equation}
(I-G)T_{\mathcal{W}\mathcal{X}} = X, \label{eq:fullMeasureBasic}
\end{equation}
based on which, identifiability of modules in $G$ can be formulated as the uniqueness of solutions for entries in $G$ given matrix $T_{\mathcal{W}\mathcal{X}}$, which is thus connected to $T_{\mathcal{W}\mathcal{X}}$'s rank and can also be tested using the following graphical rank tests. 

\begin{lemma} \label{lemma:RankPath}
Consider a network model set $\mathcal{M}$ that satisfies Assumptions~\ref{ass:InPara}, \ref{ass:fixEntry}, and let $T_{\bar{\mathcal{W}}\bar{\mathcal{X}}}(q,\theta)$ denote a submatrix of $T_{\mathcal{W}\mathcal{X}}(q,\theta)$ with its rows and columns corresponding to subsets $\bar{\mathcal{W}} \subseteq \mathcal{W}$ and $\bar{\mathcal{X}} \subseteq \mathcal{X}$, respectively. It holds that 
\begin{enumerate}
\item $rank[T_{\bar{\mathcal{W}}\bar{\mathcal{X}}}(q,\theta)] = b_{\bar{\mathcal{X}} \to \bar{\mathcal{W}}}$ generically ;

\item $rank[T_{\bar{\mathcal{W}}\bar{\mathcal{X}}}(q,\theta)] = b_{\bar{\mathcal{X}} \to \bar{\mathcal{W}}}$ globally if the set of maximum number of vertex disjoint paths from $\bar{\mathcal{X}}$ to $\bar{\mathcal{W}}$ is unique and the transfer functions contained in these paths are non-zero for all models in $\mathcal{M}$.
\end{enumerate}
\end{lemma}
\begin{proof}
The first result is proven in \cite{shi2020generic} which is extended from \cite{hendrickx2018identifiability} to the setting with known modules. The global rank has been investigated in \cite{henk2018necessary} in terms of the unique (constrained) set of vertex disjoint paths. Note that the assumption for non-zero transfer functions is implicit in \cite{henk2018necessary} as parameterized model sets are not considered there. 
\end{proof}
The above result shows that the generic and the global rank of $T_{\mathcal{W}\mathcal{X}}$ can be found by counting the maximum number of disjoint paths. In addition, we define a new notation $\bar{b}_{\bar{\mathcal{X}} \to \bar{\mathcal{W}}}$ for the global rank test in Lemma~\ref{lemma:RankPath}, i.e. the equality $\bar{b}_{\bar{\mathcal{X}} \to \bar{\mathcal{W}}}=a$ means that $b_{\bar{\mathcal{X}} \to \bar{\mathcal{W}}}=a$, and the set of maximum number of vertex disjoint paths is unique, while the transfer functions contained in those paths are non-zero for all models in the set.

In contrast to \eqref{eq:fullMeasureBasic}, when $C \not= I$ and $R \not= I$, we have $T_{\mathcal{C} \mathcal{X}}=C(I-G)^{-1}X$ instead, where $(I-G)^{-1}$ cannot be moved to the LHS to obtain a system of linear equations as in \eqref{eq:fullMeasureBasic} in general. Thus in this work, we consider identifiability of a single module in several special situations, depending on whether its input and output are measured. For each situation, identifiability conditions can still be connected to the rank of $T_{\mathcal{W}\mathcal{X}}$ and further to the graphical rank tests in Lemma~\ref{lemma:RankPath}. 

Even if each of the considered cases is limited to a specific situation and these cases cannot be combined into a single result, they together cover all the situations for single module identification in the partial measurement and partial excitation setting.

\subsection{Measured input and output} \label{subsection:MeasureInOut}
A sufficient condition is first derived for the verification of single module identifiability in the situation where both the input $w_i$ and the output $w_j$ of $G_{ji}$ are measured. 

Consider the equation $(I-G)T_{\mathcal{W}\mathcal{X}} = X$, and its $j$th row can be permuted as 
\begin{equation}
\begin{bmatrix}
- G_{ji} & -G_{j\mathcal{N}^-_j \setminus \{w_i \}} & 1 & 0 
\end{bmatrix} \begin{bmatrix}
T_{i \mathcal{X}} \\
T_{\mathcal{N}^-_j \setminus \{w_i \} \mathcal{X}} \\
T_{j \mathcal{X}}\\
\star
\end{bmatrix} = X_{j\star}, \label{eq:Measure1}
\end{equation}
where $X_{j \star}$ denotes the $j$th row vector of $X$, and $\mathcal{N}_j^-$ denotes the set of in-neighbors of $w_j$ for now and will be formally defined later. The modules contained in the $j$th row of $G$ are shown as blue blocks in Fig.~\ref{fig:FirstTheorem}.

If all internal signals are measured, i.e. $C=I$, all the submatrices of $T_{\mathcal{W}\mathcal{X}}$ in \eqref{eq:Measure1} are given by $ CT_{\mathcal{W}\mathcal{X}}$, and thus we can analyze the uniqueness for $G_{ji}$ given $T_{\mathcal{W}\mathcal{X}}$, as investigated in \cite{shi2020generic}. However, when a subset of internal signals is measured with $C$ being a selection matrix, only a subset of rows in $T_{\mathcal{W}\mathcal{X}}$ is given by $CT_{\mathcal{W}\mathcal{X}}$. Although $T_{i \mathcal{X}}$ and $T_{j \mathcal{X}}$ are available from $CT_{\mathcal{W}\mathcal{X}}$ due to measured $w_i$ and $w_j$, $T_{\mathcal{N}^-_j \setminus \{w_i \} \mathcal{X}}$ may not be directly available as the signals in $\mathcal{N}^-_j \setminus \{w_i \}$ may not be measured.

To address the unavailability of $T_{\mathcal{N}^-_j \setminus \{w_i \} \mathcal{X}}$, the following result is instrumental and can be derived from \cite[Theorem~5]{shi2020generic}. 
\begin{lemma} \label{lemma:K}
For any network model set $\mathcal{M}$ that satisfies Assumptions~\ref{ass:InPara}, \ref{ass:fixEntry} with a disconnecting set $\mathcal{D}$ from any $\bar{\mathcal{X}} \subseteq \mathcal{X}$ to any $\bar{\mathcal{W}} \subseteq \mathcal{W}$, there exists a proper transfer matrix $K(q,\theta)$ such that
\begin{equation}
T_{\bar{\mathcal{W}}  \bar{\mathcal{X}} }(q,\theta) = K(q,\theta) T_{\mathcal{D} \bar{\mathcal{X}}}(q,\theta).  \label{eq:Kresult}
\end{equation}
In addition, it holds that
\begin{itemize}
\item $K(q,\theta)$ has full column rank generically if $b_{\mathcal{D} \to \bar{\mathcal{W}}} = |\mathcal{D}|$;

\item $K(q,\theta)$ has full column rank globally if $\bar{b}_{\mathcal{D} \to \bar{\mathcal{W}}_1} = |\mathcal{D}|$ for some $\bar{\mathcal{W}}_1 \subseteq \bar{\mathcal{W}}$.
\end{itemize} 
\end{lemma}

The above result shows that if an appropriate disconnecting set $\mathcal{D}$ is chosen, as in Fig.~\ref{fig:FirstTheorem}, $T_{\mathcal{N}^-_j \setminus \{w_i \} \mathcal{X}}$ can be factorized into $K T_{\mathcal{D}\mathcal{X}}$ for some $K$, where $T_{\mathcal{D}\mathcal{X}}$ can be obtained from $C T_{\mathcal{W}\mathcal{X}}$ in certain way, e.g., the signals in $\mathcal{D}$ are measured. This together with \eqref{eq:Measure1} leads to 
\begin{equation}
\begin{bmatrix}
G_{ji} & G_{j\mathcal{W}^-_j \setminus \{w_i \}} K
\end{bmatrix} \begin{bmatrix}
T_{i \mathcal{X}} \\
T_{\mathcal{D} \mathcal{X}}
\end{bmatrix} = T_{j \mathcal{X}}-X_{j\star}, \label{eq:sufficientReasoning}
\end{equation}
and thus the uniqueness of $G_{ji}$ is ensured if $\begin{bmatrix}
T_{i \mathcal{X}} \\
T_{\mathcal{D} \mathcal{X}}
\end{bmatrix}$ has full row rank and the signals in $\mathcal{D}\cup \{w_i,w_j\}$ are measured. 

The above requirement for full row rank can be reformulated into a path-based condition using Lemma~\ref{lemma:RankPath}, i.e. $\mathcal{D}$ is excited or indirectly excited by a set $\bar{\mathcal{X}}$ of external signals as in Fig.~\ref{fig:FirstTheorem}. In addition, the requirement for measuring $\mathcal{D}$ can be further relaxed by the indirect measurement of $\mathcal{D}$, i.e. the signals in $\mathcal{D}$ are not necessarily measured but have paths to a set $\bar{\mathcal{C}}$ of measured internal signals as in Fig.~\ref{fig:FirstTheorem}.

\begin{figure}[h]
\centering
\includegraphics[scale=0.45]{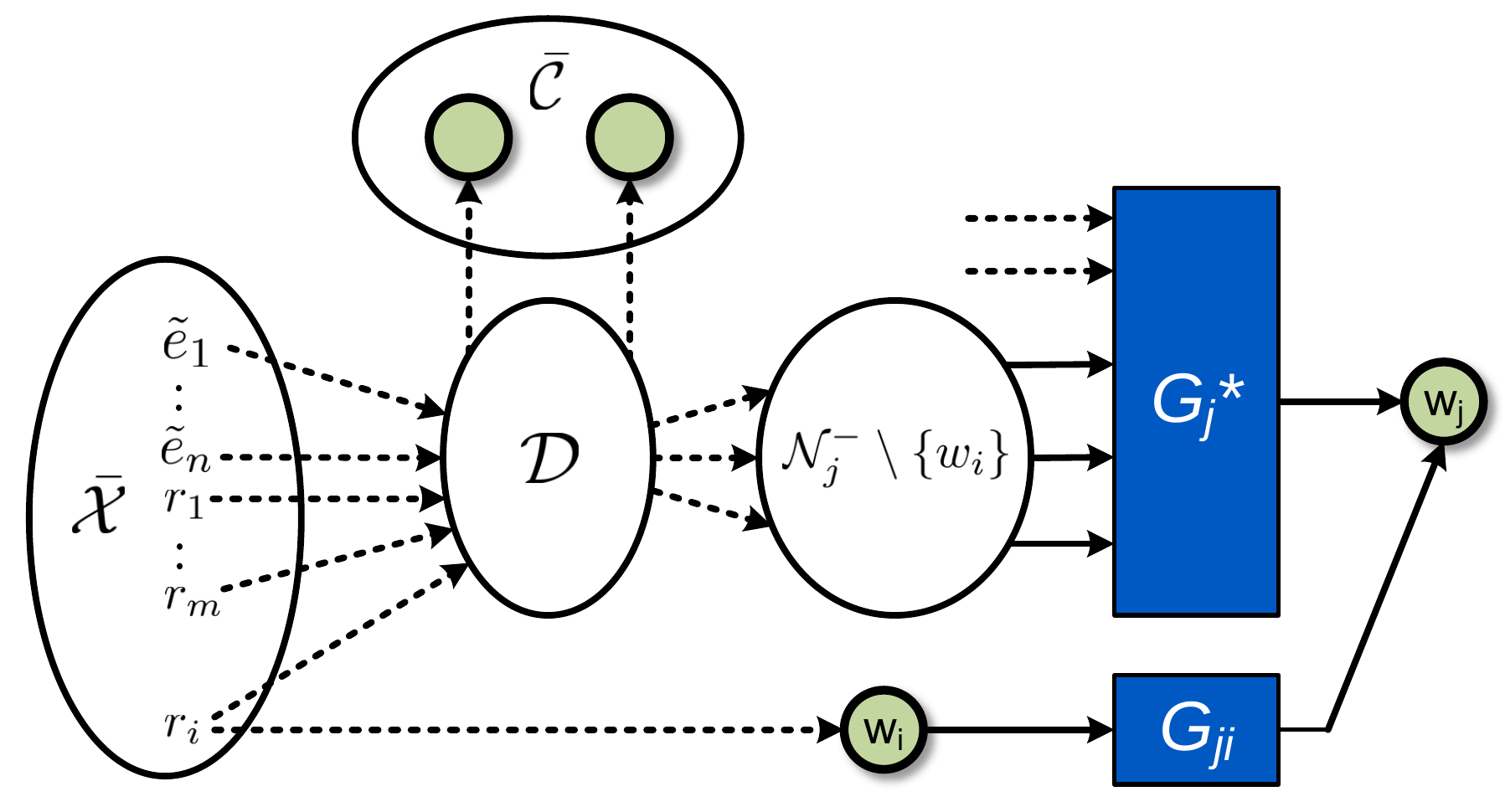}
\caption{Visualization of a situation where $G_{ji}$ is generically identifiable. $G_j^*$ represents the other in-coming modules of $w_j$, and the dashed edges represent directed paths.}
\label{fig:FirstTheorem}
\end{figure}

The above reasoning for identifiability analysis can be generalized. Before introducing this result, we first define an important set of signals:
\begin{itemize}
\item Let set $\mathcal{N}^-_j$ contain all internal signals that are in-neighbors of $w_j$ excluding the measured signals that have a known directed edge to $w_j$.
\end{itemize}
According to the above definition, it holds that $\mathcal{W}_j^- \subseteq \mathcal{N}^-_j$. When all the non-zero modules are unknown, we have $\mathcal{N}^-_j = \mathcal{W}_j^-$ which simply contains all the in-coming internal signals of $w_j$.

Then the following graphical result can be obtained from the generalization of the reasoning in \eqref{eq:sufficientReasoning}.

\begin{theorem} \label{theorem:partial}
For a model set $\mathcal{M}$ that satisfies Assumptions~\ref{ass:fullrankPhi}, \ref{ass:InPara}, \ref{ass:FeedThrough}, \ref{ass:fixEntry} with its graph, $G_{ji}(q,\theta)$ is generically identifiable in $\mathcal{M}$ from $(w_\mathcal{C},r)$ if for some $\bar{\mathcal{X}} \subseteq \mathcal{X}_j$ and $\bar{\mathcal{C}} \subseteq \mathcal{C} \setminus \{w_i\}$, there exists a $\bar{\mathcal{X}} - (\mathcal{N}^-_j\setminus \{w_i\} ) \cup \bar{\mathcal{C}}$ disconnecting set $\mathcal{D}$ such that
\begin{enumerate}
\item \label{partial1:con1} $b_{\bar{\mathcal{X}} \to \{w_i\} \cup \mathcal{D}} = |\mathcal{D}| + 1$;

\item  \label{partial1:con2} $b_{\mathcal{D} \to \bar{\mathcal{C}}} = |\mathcal{D}|$.

\item $w_i$ and $w_j$ are in $\mathcal{C}$.
\end{enumerate}
\end{theorem}

The above result is visualized in Fig.~\ref{fig:FirstTheorem} for the special case where $w_j \notin \bar{\mathcal{C}}$, for simplicity. It shows that to identify $G_{ji}$, instead of measuring and exciting all the inputs of the MISO subsystem that contains $G_{ji}$, we only need to excite and measure the signals in $\mathcal{D} \cup \{w_i\}$ indirectly. The difficulty in applying the above result may arise from the need to search for the subsets $\bar{\mathcal{X}}$ and $\bar{\mathcal{C}}$, which, however, cannot be avoided due to the coupling between the excitation signals and the measured signals that are relevant to identifiability of $G_{ji}$.

Compared to \cite[Theorem~4]{shi2020generic} which states that the signals in $\{w_i\} \cup \mathcal{D}$ need to be excited when all the internal signals are measured, the above result is a generalization which only requires the signals in $\{w_i, w_j\} \cup \mathcal{D}$ to be measured. In addition, $\mathcal{D}$ can also be measured indirectly as in condition~\eqref{partial1:con2}, i.e. the signals in $\mathcal{D}$ are not measured but have vertex disjoint paths to measured signals in $\bar{\mathcal{C}}$. 

This indirect measurement of $\mathcal{D}$ can be shown to appear also in the network identification method of \cite{ramaswamyCDC2019}. For the consistent estimate of $G_{ji}$, the method requires the indirect measurement of signals that block the so-called parallel paths from $w_i$ to $w_j$ and the loops around $w_j$, while these signals actually coincide with $\mathcal{D}$ as shown in \cite{shi2020generic}. Thus, the considered experimental setup in \cite{ramaswamyCDC2019} matches the one in Theorem~\ref{theorem:partial}. 

\begin{remark} \label{remark:globalPath}
Based on the connection between the unique set of vertex disjoint paths and the global rank of transfer matrices as in Lemma~\ref{lemma:RankPath}, Theorem~\ref{theorem:partial} can be modified to address global identifiability by considering $\bar{b}_{\bar{\mathcal{X}} \to \{w_i\} \cup \mathcal{D}}$ and $\bar{b}_{\mathcal{D} \to \bar{\mathcal{C}}}$ instead.
\end{remark}

Theorem~\ref{theorem:partial} has a potential application for signal and sensor allocation, as it explicitly states that signals in $\{w_i\} \cup \mathcal{D}$ should be excited or indirectly excited as in condition~\eqref{partial1:con1}, and the signals in $\mathcal{D}$ should be measured or indirectly measured as in condition~\eqref{partial1:con2}. However, it can be difficult to perform an identifiability test for a given model set as one needs to search for a disconnecting set. In order to better facilitate such an analysis step, an equivalent path-based condition of Theorem~\ref{theorem:partial} is developed.

\begin{proposition} \label{proposition:partialPath}
For a model set $\mathcal{M}$ that satisfies Assumptions~\ref{ass:fullrankPhi}, \ref{ass:InPara}, \ref{ass:FeedThrough}, \ref{ass:fixEntry} with its graph, $G_{ji}(q,\theta)$ is generically identifiable in $\mathcal{M}$ from $(w_\mathcal{C},r)$ if for some $\bar{\mathcal{X}} \subseteq \mathcal{X}_j$ and $\bar{\mathcal{C}} \subseteq \mathcal{C} \setminus \{w_i\}$, it holds that
\begin{enumerate}
\item \label{partial1Path:con1} $b_{\bar{\mathcal{X}} \to  \mathcal{N}^-_j \cup \bar{\mathcal{C}}} = b_{\bar{\mathcal{X}} \to (\mathcal{N}^-_j\setminus \{w_i\} ) \cup \bar{\mathcal{C}}} + 1$;

\item \label{partial1Path:con2} $b_{\bar{\mathcal{X}} \to (\mathcal{N}^-_j\setminus \{w_i\} ) \cup \bar{\mathcal{C}}} =b_{\bar{\mathcal{X}} \to \bar{\mathcal{C}}}$.

\item $w_i$ and $w_j$ are in $\mathcal{C}$.
\end{enumerate}
\end{proposition}

Compared to Theorem~\ref{theorem:partial}, the above result avoids the search for a disconnecting set and thus is easier for analyzing identifiability. However, it is less informative than Theorem~\ref{theorem:partial} since it does not specify explicitly where to allocate excitation signals and sensors for single module identifiability.  

The conditions in this subsection are illustrated in the following example.

\begin{example}
Consider a model set $\mathcal{M}$ in Fig.~\ref{fig:exam1} with $G_{21}$ of interest and the measured internal signals $\mathcal{C}= \{w_1,w_2,w_3, w_6\}$. It can be found that $\mathcal{N}_2^-=\{w_1,w_4\}$, which does not include the in-neighbor $w_3$ of $w_2$ since $w_3$ is measured and has a known edge to $w_2$. In addition, we have $\mathcal{X}_2=\{\tilde{e_1},r_4,r_5\}$ because these external signals do not have an unknown edge to $w_2$. The goal is then to verify generic and global identifiability of $G_{21}$ using the graphical conditions.

\begin{figure}[h]
\centering
\includegraphics[scale=0.35]{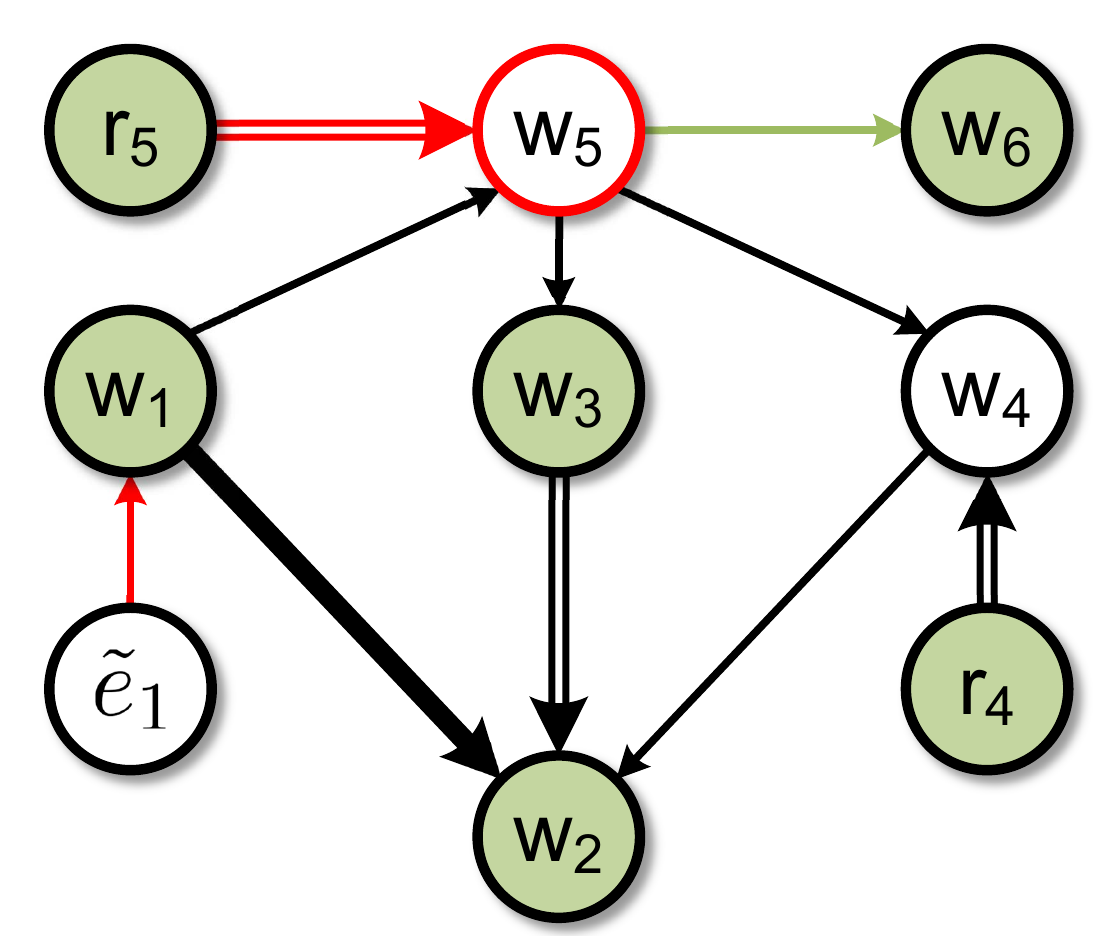}
\caption{Identifiable $G_{21}$ (thick edge) with only green internal signals measured and one known module $G_{23}$ (double-line edge).}
\label{fig:exam1}
\end{figure}

By taking $\bar{\mathcal{X}} = \{\tilde{e}_1,r_5\}$ and $\bar{\mathcal{C}} = \{w_6\}$, it holds that $\{w_5\} = \mathcal{D}$ is a set that disconnects $\bar{\mathcal{X}}$ from $(\mathcal{N}_2^- \setminus\{w_1\})\cup \bar{\mathcal{C}} = \{w_4,w_6\} $, as indicated by a red vertex in Fig.~\ref{fig:exam1}. Thus, condition~\eqref{partial1:con1} of Theorem~\ref{theorem:partial} is satisfied as there are two vertex disjoint paths $\tilde{e}_1 \to w_1$ and $r_5 \to w_5$ indicated by the red arrows in Fig.\ref{fig:exam1}. In addition, condition~\eqref{partial1:con2} also holds because of the green path $w_5 \to w_6$. Then based on Theorem~\ref{theorem:partial}, $G_{21}$ is generically identifiable. The conditions in Proposition~\ref{proposition:partialPath} can be verified similarly with the chosen $\bar{\mathcal{X}}$ and $\bar{\mathcal{C}}$.

If the transfer functions are non-zero in all models of the model set, $G_{21}$ is also globally identifiable as the set of maximum number of vertex disjoint paths from $\{\tilde{e},r_5\}$ to $\{w_1,w_5\}$ and the one from $w_5$ to $w_6$ are unique.
\end{example}

Next, we extend the result in Proposition~\ref{proposition:partialPath} from a single module to a subnetwork, i.e. a subset of in-coming modules of $w_j$.

\begin{corollary} \label{corollary:partialPathSubnetwork} 
For a model set $\mathcal{M}$ that satisfies Assumptions~\ref{ass:fullrankPhi}, \ref{ass:InPara}, \ref{ass:FeedThrough}, \ref{ass:fixEntry} with its graph, let $G_{j \mathcal{N}^\star}(q,\theta)$, with any $\mathcal{N}^\star \subseteq \mathcal{N}_j^- $, denote a vector of unknown modules from the $j$th row of $G(q,\theta)$. $G_{j\mathcal{N}^\star}(q,\theta)$ is generically identifiable in $\mathcal{M}$ from $(w_\mathcal{C},r)$ if for some $\bar{\mathcal{X}} \subseteq \mathcal{X}_j$ and $\bar{\mathcal{C}} \subseteq \mathcal{C} \setminus \mathcal{N}^\star $, it holds that
\begin{enumerate}
\item $b_{\bar{\mathcal{X}} \to  \mathcal{N}^-_j \cup \bar{\mathcal{C}}} = b_{\bar{\mathcal{X}} \to (\mathcal{N}^-_j\setminus \mathcal{N}^\star)  \cup \bar{\mathcal{C}}} + b_{\bar{\mathcal{X}} \to \mathcal{N}^\star}$ and \\ $b_{\bar{\mathcal{X}} \to \mathcal{N}^\star} = |\mathcal{N}^\star |$;

\item  $b_{\bar{\mathcal{X}} \to (\mathcal{N}^-_j\setminus \mathcal{N}^\star) \cup \bar{\mathcal{C}}} =b_{\bar{\mathcal{X}} \to \bar{\mathcal{C}}}$.

\item $\mathcal{N}^\star$ and $w_j$ are in $\mathcal{C}$.
\end{enumerate}
\end{corollary}
\begin{proof}
The proof is analogous to the proof of Proposition~\ref{proposition:partialPath} and can be based on the proof of Theorem~\ref{theorem:partial}. In this case, we only need to replace $G_{ji}$ with $G_{j \mathcal{N}^\star}$ and $T_{i \bar{\mathcal{X}}}$ with $T_{\mathcal{N}^\star \bar{\mathcal{X}}}$ in \eqref{eq:proofPar1}. Then the rank condition implied by the first path-based condition leads to a unique solution for $G_{j \mathcal{N}^\star}$.
\end{proof}

Corollary~\ref{corollary:partialPathSubnetwork} is related to \cite[Theorem~IV.4]{bazanella2019network} which also specifies sufficient conditions for generic identifiability of $G_{ji}$, in a setting without known non-zero modules and with only $r$ signals as excitation sources for identifiability analysis. In addition, the theorem not only contains a graphical condition, but also requires the prior knowledge for certain submatrices of $T_{\mathcal{W}\mathcal{X}}$. When all non-zero modules are unknown and only $r$ signals are considered as excitation sources for identifiability analysis, the graphical condition there is equivalent to the first condition in Corollary~\ref{corollary:partialPathSubnetwork}, while the corollary further specifies the graphical conditions, under which the required submatrices of \cite[Theorem~IV.4]{bazanella2019network} can be obtained. 

\subsection{Excited input and output} \label{sec:ExciInAndOut}
In the previous section, it is assumed that both the input and the output of a module are measured, which may not be feasible in some practical situations. This motivates us to consider the situation where the input or the output can be unmeasured, with the cost that they are excited.

In this case, instead of starting with the equation $(I-G)T_{\mathcal{W}\mathcal{X}}=X$ as in \eqref{eq:Measure1}, we analyze identifiability using the $i$-th column of $C=T_{\mathcal{C}\mathcal{W}}(I-G)$, obtained from the relation $T_{\mathcal{C}\mathcal{W}} = CT$. Unlike \eqref{eq:Measure1} which treats $G_{ji}$ in the corresponding MISO subsystem, we take a dual perspective now and analyze $G_{ji}$ in the corresponding SIMO subsystem, as shown in Fig.~\ref{fig:2ndTheorem} where the two blue blocks denote the module $G_{ji}$ and the other out-going modules of $w_i$. Thus, a dual result of Lemma~\ref{lemma:K} can be obtained analogously.

\begin{lemma} \label{lemma:Kdual} 
For any network model set $\mathcal{M}$ that satisfies Assumptions~\ref{ass:InPara}, \ref{ass:fixEntry} with a disconnecting set $\mathcal{D} $ from any $\bar{\mathcal{X}} \subseteq \mathcal{X}$ to any $\bar{\mathcal{W}} \subseteq \mathcal{W}$, there exists a proper transfer matrix $K(q,\theta)$ such that
\begin{equation}
T_{\bar{\mathcal{W}}  \bar{\mathcal{X}} }(q,\theta) = T_{\bar{\mathcal{W}}  \mathcal{D}}(q,\theta)K(q,\theta). \label{eq:Kdual}
\end{equation}
Additionally, it holds that
\begin{itemize}
\item $K(q,\theta)$ has full row rank generically if $b_{\bar{\mathcal{X}} \to \mathcal{D} } = |\mathcal{D}|$;

\item $K(q,\theta)$ has full row rank globally if $\bar{b}_{\bar{\mathcal{X}}_1 \to \mathcal{D} } = |\mathcal{D}|$ for some $\bar{\mathcal{X}}_1 \subseteq \bar{\mathcal{X}}$.
\end{itemize} 
\end{lemma}

In view of analyzing generic identifiability of $G_{ji}$ when its input and output are excited, the above result allows us to find an appropriate disconnecting set $\mathcal{D}$, as shown in Fig.~\ref{fig:2ndTheorem}. It can then be found that generic identifiability of $G_{ji}$ is satisfied, if the signals in $\mathcal{D}$ are either excited or indirectly excited by a set $\bar{\mathcal{X}}$ of external signals and either measured or indirectly measured. This is illustrated in Fig.~\ref{fig:2ndTheorem} and can be formalized as follows.

\begin{figure}[h]
\centering
\includegraphics[scale=0.45]{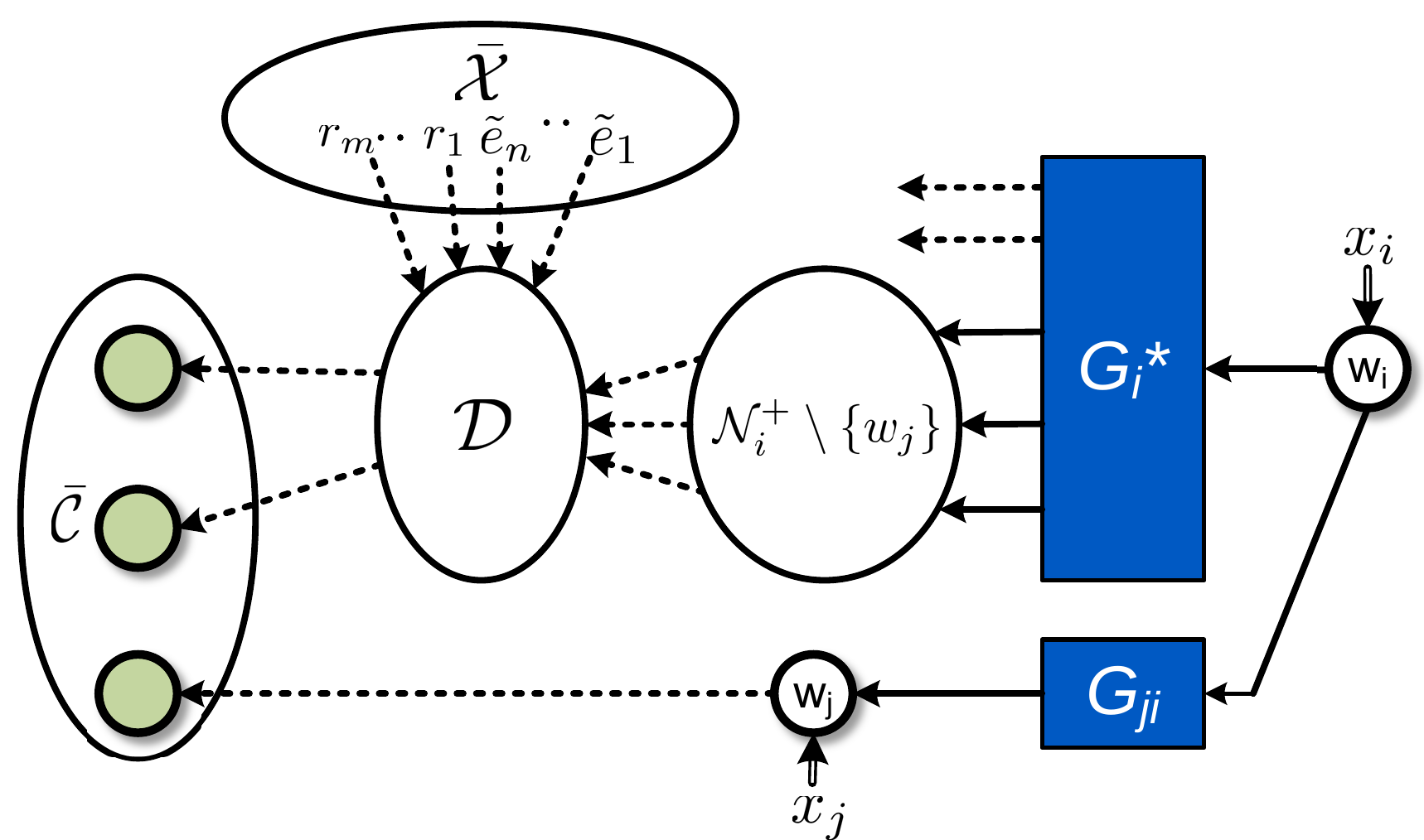}
\caption{Visualization of a situation where $G_{ji}$ is generically identifiable, and its input and output are excited. $G_i^*$ represents the other out-going modules of $w_i$, and the dashed edges represent directed paths.}
\label{fig:2ndTheorem}
\end{figure}

To introduce the identifiability results, we first define an important set related to the input $w_i$ of $G_{ji}$:
\begin{itemize}
\item Let $\mathcal{N}_i^+$ contain the out-neighbors of $w_i$ excluding the ones that satisfy  (i) $w_i$ has a known directed edge to it, \textit{and}  (ii) it is excited by an external signal that has out-degree one with a known out-going edge.
\end{itemize}
Therefore, it holds that $\mathcal{W}_i^+ \subseteq \mathcal{N}_i^+$, and when all modules are unknown, we have $\mathcal{N}_i^+ = \mathcal{W}_i^+$ which simply contains all the out-neighbors of $w_i$. 

Then the following result can be obtained.

\begin{theorem} \label{theorem:partial2}
For a model set $\mathcal{M}$ that satisfies Assumptions~\ref{ass:fullrankPhi}, \ref{ass:InPara}, \ref{ass:FeedThrough}, \ref{ass:fixEntry}, let there be two vertices $x_i$ and $x_j$ in $\mathcal{X}$, which have out-degree $1$ and known directed edges to $w_i$ and $w_j$ respectively. Then  
$G_{ji}(q,\theta)$ is generically identifiable in $\mathcal{M}$ from $(w_\mathcal{C},r)$ if 
for some $\bar{\mathcal{X}} \subseteq \mathcal{X}_j \setminus \{x_j\}$ and $\bar{\mathcal{C}}\subseteq \mathcal{C}$, there exists a $\mathcal{N}_i^+ \setminus \{w_j\} \cup \bar{\mathcal{X}} - \bar{\mathcal{C}}$ disconnecting set $\mathcal{D}$ such that
\begin{enumerate}
\item \label{partial2:con1} $b_{ \{w_j\} \cup \mathcal{D}  \to \bar{\mathcal{C}}} = |\mathcal{D}| + 1;$

\item \label{partial2:con2} $b_{\bar{\mathcal{X}} \to \mathcal{D}}=|\mathcal{D}|$.
\end{enumerate}
\end{theorem}
\begin{proof}
The proof is analogous to Theorem~\ref{theorem:partial} by considering the $i$-th column of $C=T_{\mathcal{C}\mathcal{W}}(I-G)$. Note that the analysis requires the availability of $T_{\mathcal{C}j}$ and $T_{\mathcal{C}i}$, where $T_{\mathcal{C}i}$ can be obtained from $CTX$ when $w_i$ is excited by $x_i$. $x_i$ can be either a measured signal $r_i$ or a noise signal $\tilde{e}_i$: Firstly, $r_i$ always has out-degree one, and it is clear that $T_{\mathcal{C}i}$ is a submatrix of $CTR$. On the other hand, when $\tilde{e}_i$ has out-degree one and a known edge to $w_i$, its corresponding column $H_{\star i}$ in $H$ has only one non-zero entry that is also known. Therefore, $T_{\mathcal{C}i}$ can be obtained from $C T H_{\star i}$ as well. The mapping $T_{\mathcal{C}j}$ can be obtained from $CTX$ similarly due to the existence of $x_j$.
\end{proof}
Theorem~\ref{theorem:partial2} is visualized in Fig~\ref{fig:2ndTheorem}. It shows that to identify $G_{ji}$, instead of measuring and exciting all the outputs of the SIMO subsystem that contains $G_{ji}$, it is sufficient to measure and excite the signals in $\mathcal{D} \cup \{w_j\}$ indirectly.

Theorem~\ref{theorem:partial2} is also a dual result of Theorem~\ref{theorem:partial}: while Theorem~\ref{theorem:partial2} considers the input and the output to be excited, Theorem~\ref{theorem:partial} assumes them to be measured; in addition, the graphical conditions of the two results have a similar dual structure. The result can also be extended to address global identifiability by requiring that the sets of maximum number of vertex disjoint paths are unique. In addition, a path-based formulation of Theorem~\ref{theorem:partial2} can also be obtained, which is analogous to Proposition~\ref{proposition:partialPath}.

\begin{proposition} \label{proposition:partialPath2}
For a model set $\mathcal{M}$ that satisfies Assumptions~\ref{ass:fullrankPhi}, \ref{ass:InPara}, \ref{ass:FeedThrough}, \ref{ass:fixEntry}, let there be two vertices $x_i$ and $x_j$ in $\mathcal{X}$, which have out-degree $1$ and known directed edges to $w_i$ and $w_j$ respectively. Then $G_{ji}(q,\theta)$ is generically identifiable in $\mathcal{M}$ from $(w_\mathcal{C},r)$ if for some $\bar{\mathcal{X}} \subseteq \mathcal{X}_j \setminus \{x_j\}$ and $\bar{\mathcal{C}}\subseteq \mathcal{C}$, it holds that
\begin{enumerate}
\item  $b_{\mathcal{N}_i^+  \cup \bar{\mathcal{X}} \to \bar{\mathcal{C}}} = b_{\mathcal{N}_i^+  \setminus \{w_j\} \cup \bar{\mathcal{X}} \to \bar{\mathcal{C}}} + 1$;

\item  $b_{\mathcal{N}_i^+  \setminus \{w_j\} \cup \bar{\mathcal{X}} \to \bar{\mathcal{C}}} = b_{\bar{\mathcal{X}} \to \bar{\mathcal{C}}}$.
\end{enumerate}
\end{proposition}

While Proposition~\ref{proposition:partialPath2} is a dual result of Proposition~\ref{proposition:partialPath}, it allows us to analyze identifiability in a completely different setting, as illustrated in the following example.

\begin{example}
Consider identifiability of $G_{21}$ in $\mathcal{M}$ of Fig.~\ref{fig:Example2} where $\{w_4,w_5,w_7\}$ contains the measured signals, and $\{w_1,w_2\}$ is excited by $\{r_1,\tilde{e}_2\}$. It holds that $\mathcal{N}_1^+ = \{w_2,w_3,w_5\}$, and furthermore, we can choose $\bar{\mathcal{X}} = \{\tilde{e}_5, r_6\}$ and $\bar{\mathcal{C}} = \{w_4,w_5,w_7\}$.

\begin{figure}[h]
\centering
\includegraphics[scale=0.3]{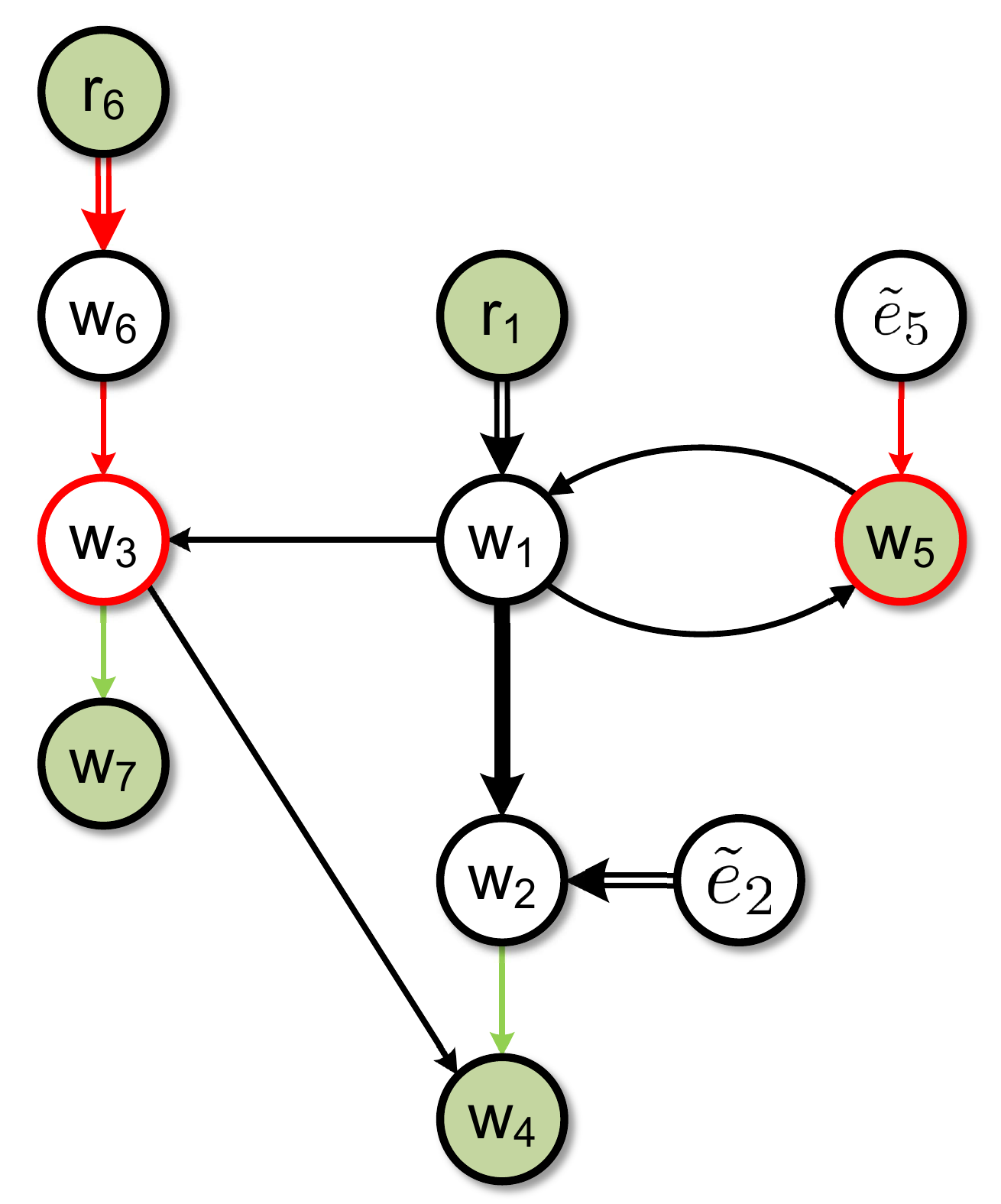}
\caption{Identifiability of $G_{21}$ is considered (thick line) in $\mathcal{M}$ with both input and output unmeasured, while $w_1$ and $w_2$ are excited by $r_1$ and $\tilde{e}_2$ via known edges (double-line edges), respectively.}
\label{fig:Example2}
\end{figure}

It can be found that $\mathcal{D} = \{w_3,w_5\}$ is disconnecting $(\mathcal{N}_1^+  \setminus \{w_2\}) \cup\bar{\mathcal{X}} = \{w_3,w_5,\tilde{e}_5, r_6\}$ from $\bar{\mathcal{C}} $, as indicated by the red vertices in Fig.~\ref{fig:Example2}. In addition, condition~\eqref{partial2:con1} of Theorem~\ref{theorem:partial2} is satisfied since there are three vertex disjoint paths including $w_3 \to w_7$, $w_2 \to w_4$, and $w_5$ itself, indicated by green edges in Fig.~\ref{fig:Example2}. Condition~\eqref{partial2:con2} is then satisfied because there are two vertex disjoint paths indicated by the red edges, including $r_6 \to w_3$ and $\tilde{e}_5 \to w_5$, which concludes generic identifiability of $G_{21}$. Note that $r_6$ excites $w_3$ indirectly through a path.

Since the above sets of maximum number of vertex disjoint paths are unique, $G_{21}$ is also globally identifiable.
\end{example}

\begin{remark} \label{remark:SIMO}
When identifiability of multiple modules $G_{\mathcal{N}^\star i}$ from one SIMO subsystem is considered, where $\mathcal{N}^\star \subseteq \mathcal{N}_i^+$, the results in this subsection can be extended in a straightforward way by considering $\mathcal{N}^\star$ instead of the single output $w_j$, which is similar to the extension in Corollary~\ref{corollary:partialPathSubnetwork}.
\end{remark}

The above remark leads to a graphical result that is related to \cite[Theorem~IV.2]{bazanella2019network}, which also considers identifiability of multiple modules from one SIMO model. However, \cite[Theorem~IV.2]{bazanella2019network} is not fully graphical as it requires the availability of certain mappings, and it does not consider known modules nor the excitation contributed by unmeasured noises.

\section{Sufficient conditions: measured input or output with indirect excitation}\label{sec:Un-inputorUn-output}
Even if Proposition~\ref{proposition:partialPath2} considers the most general measurement scheme, it requires $w_i$ and $w_j$ to be excited. When $w_i$ (or $w_j$) is measured, it is not necessary to excite $w_i$ (or $w_j$), as shown in the last condition of Lemma~\ref{lemma:nece1}. In this section, we develop graphical identifiability conditions for the situation where either $w_i$ or $w_j$ is measured and not excited.

We first consider the case where the input $w_i$ is unmeasured and the output $w_j$ is measured, and the measured $w_j$ may not be excited by $r$ or $\tilde{e}$. As in \eqref{eq:Measure1}, the mapping $T_{i \mathcal{X}}$ is not available from $CTX$ as $w_i$ is unmeasured and thus needs to be represented by available mappings via an appropriately chosen disconnecting set, which can be motivated by the following example:
\begin{example} \label{example:input}
Consider the network model set in Fig.~\ref{fig:examInput}(a) where identifiability of $G_{21}$ is of interest while $w_1$ is unmeasured. Since the mapping from $r_1$ to $w_1$ is known to be $1$, $G_{21}$ can be uniquely recovered from the available external-to-internal mapping from $r_1$ to $w_2$.
\begin{figure}[h]
\begin{minipage}{0.23\textwidth}
\centering
\includegraphics[scale=0.34]{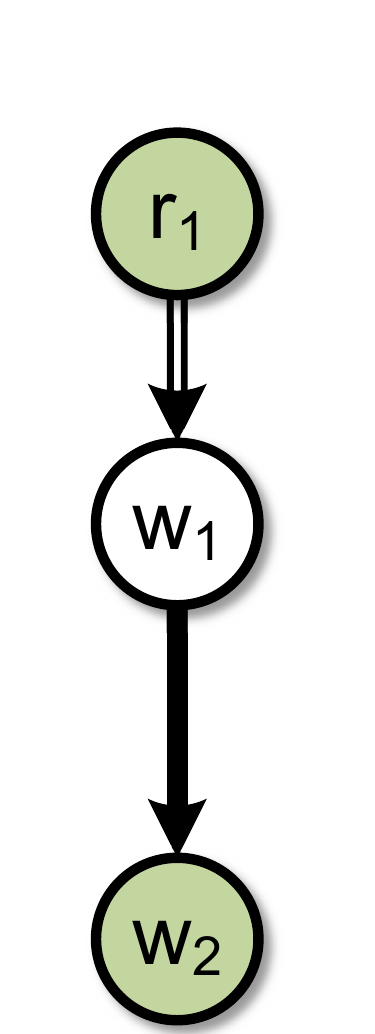}
\\(a)
\end{minipage}
\begin{minipage}{0.23\textwidth}
\centering
\includegraphics[scale=0.34]{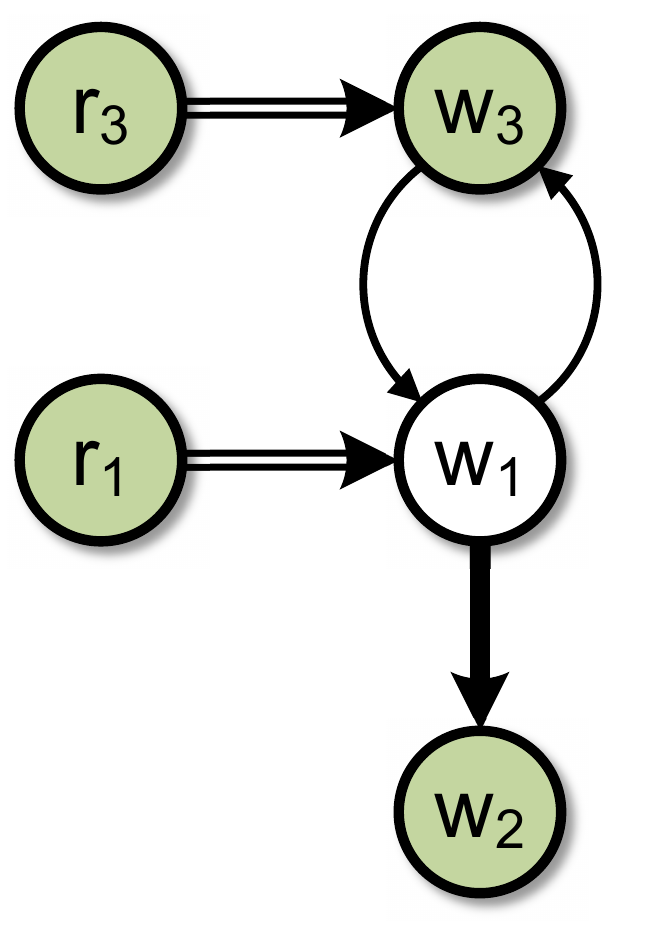}
\\(b)
\end{minipage}
\caption{Two network model sets with $G_{21}$ as the target module and $w_1$ unmeasured. $G_{21}$ is generically and globally identifiable in both cases under measured signals (green).}
\label{fig:examInput}
\end{figure}
In Fig.~\ref{fig:examInput}(b) there is a loop around $w_1$, and the mapping from $r_1$ to $w_2$ is 
$$
T_{w_2r_1} = \frac{G_{21} }{1-G_{13}G_{31}}, 
$$
and thus $G_{21}$ cannot be recovered from $T_{w_2r_1}$ alone as in Fig.~\ref{fig:examInput}(a). However, the loop transfer can be found as 
$$
T_{w_3r_3} =  \frac{1 }{1-G_{13}G_{31}},
$$
and thus 
$
G_{21}=T_{w_2r_1}(T_{w_3r_2})^{-1},
$
where both mappings are available because $w_3$ and $w_2$ are measured. $\hfill{} \blacksquare$
\end{example}

As shown in the above example, it is important to measure and excite the vertices in the loops around the unmeasured input, in order to achieve identifiability of the module under consideration. This observation can be generalized as follows.

\begin{theorem} \label{theorem:UnmeasureInput}
For a model set $\mathcal{M}$ that satisfies Assumptions~\ref{ass:fullrankPhi}, \ref{ass:InPara}, \ref{ass:FeedThrough}, \ref{ass:fixEntry}, let set $\mathcal{N}_i^ *$ contain all the in-neighbors of $w_i$ in $\mathcal{W}$. Suppose that $w_j$ is measured but $w_i$ cannot be measured, and let there be a vertex $x_i$ in $\mathcal{X}$, which has out-degree $1$ and a known edge to $w_i$. Then module $G_{ji}$ is generically identifiable in $\mathcal{M}$ from $(w_\mathcal{C},r)$ if for some $\bar{\mathcal{X}} \subseteq \mathcal{X}_j \setminus \{x_i\}$  and $\bar{\mathcal{C}} \subseteq \mathcal{C}$, there exists a $\bar{\mathcal{X}} \cup \{x_i\}  - \mathcal{N}_i^* \cup (\mathcal{N}_j^- \setminus \{w_i\}  ) \cup \bar{\mathcal{C}}$ disconnecting set $\mathcal{D}$ such that
\begin{enumerate}
\item $b_{\bar{\mathcal{X}} \cup \{x_i\} \to \mathcal{D} \cup \{w_i\} } = 1+ |\mathcal{D}|$;

\item $b_{\mathcal{D} \to \bar{\mathcal{C}}} = |\mathcal{D}|$.
\end{enumerate}
\end{theorem}

In the above result, the disconnecting set intersects with the paths from $x_i$ to $\mathcal{N}_i^*$, which implies that all loops around $w_i$ are blocked by the disconnecting set $\mathcal{D}$, matching the observation from Example~\ref{example:input}. Therefore, this result is an extension of Theorem~\ref{theorem:partial} to the setting with unmeasured input, by additionally blocking all the loops around the unmeasured input.

In addition, both Theorem~\ref{theorem:partial2} and Theorem~\ref{theorem:UnmeasureInput} can be used to analyze identifiability of $G_{ji}$ with unmeasured input and measured output, while Theorem~\ref{theorem:UnmeasureInput} provides the extra freedom that $w_j$ does not need to be excited.

Then the corresponding path-based formulation of Theorem~\ref{theorem:UnmeasureInput} can also be derived.
\begin{corollary} \label{corollary:UnmeasureInputPath}
For a model set $\mathcal{M}$ that satisfies Assumptions~\ref{ass:fullrankPhi}, \ref{ass:InPara}, \ref{ass:FeedThrough}, \ref{ass:fixEntry} with its graph, suppose that $w_j$ is measured but $w_i$ cannot be measured, and let there be a vertex $x_i$ in $\mathcal{X}$, which has out-degree $1$ and a known edge to $w_i$. Then module $G_{ji}$ is generically identifiable if for some $\bar{\mathcal{X}} \subseteq \mathcal{X}_j \setminus \{x_i\}$ and $\bar{\mathcal{C}} \subseteq \mathcal{C}$, the following conditions hold:
\begin{enumerate}
\item $b_{\bar{\mathcal{X}}\cup \{x_i\} \to \mathcal{N}_i^* \cup \mathcal{N}_j \cup \bar{\mathcal{C}}} = 1+ b_{\bar{\mathcal{X}}\cup \{x_i\} \to \mathcal{N}_i^* \cup(\mathcal{N}_j^- \setminus \{w_i\}  ) \cup \bar{\mathcal{C}}}$;

\item $b_{\bar{\mathcal{X}}\cup \{x_i\} \to \mathcal{N}_i^* \cup(\mathcal{N}_j^- \setminus \{w_i\}  ) \cup \bar{\mathcal{C}}} = b_{\bar{\mathcal{X}} \cup \{x_i\}  \to  \bar{\mathcal{C}}}$.
\end{enumerate}
\end{corollary}

The above result can also be extended trivially to a subset of in-coming modules of $w_j$, which is similar to the extension in Corollary~\ref{corollary:partialPathSubnetwork}.

For the situation where the input $w_i$ is measured while the output $w_j$ is unmeasured, the requirement for direct excitation for $w_i$ in Theorem~\ref{theorem:partial2} can also be relaxed. Here we start with the analysis in a SIMO problem as in Section~\ref{sec:ExciInAndOut}. When the direct excitation for $w_i$ is not present, the mapping $T_{\mathcal{C}i}$ is thus not given by $CT_{\mathcal{W}\mathcal{X}}$. However, based on Lemma~\ref{lemma:Kdual}, we can simply require a disconnecting set $\mathcal{D}$ from $\{w_i\}$ to $\mathcal{C}$, and then $T_{\mathcal{C}i}$ is represented by $T_{\mathcal{C}\mathcal{D}}K$ for some transfer matrix $K$ where $T_{\mathcal{C}\mathcal{D}}$ can be obtained from other mappings when the signals in $\mathcal{D}$ are excited or indirectly excited.

\begin{table*}
\caption{Different synthesis approaches for identifiability of $G_{ji}$.}
\label{table:1}
  \centering
  \begin{tabular}{ |p{1.5cm}|p{3.4cm}|p{3.4cm}|p{3.4cm}|p{3.4cm}|  }
 \hline
  & Case 1: Measured $w_i$ and $w_j$ & Case 2: Excited $w_i$ and $w_j$  &Case 3: unmeasured $w_i$ and measured $w_j$   &  Case 4: measured $w_i$ and unmeasured $w_j$    \\ 
 \hline
 Disconnecting set $\mathcal{D}$ & $\{w_i\} - \mathcal{N}^-_j \setminus \{w_i\}$ with $w_i \notin \mathcal{D}$ & $\mathcal{N}_i^+ \setminus \{w_j\} - \{w_j\}$ with $w_j \notin \mathcal{D}$ & $\{w_i\} - \mathcal{N}_i^* \cup (\mathcal{N}_j^- \setminus \{w_i\}  )$ with $w_i \notin \mathcal{D}$  &$(\mathcal{N}_i^+ \setminus \{w_j\}) \cup \{w_i\} - \{w_i,w_j\}$ with $w_j \notin \mathcal{D}$  \\ 
\hline 
 Excitation allocation & Excite $\mathcal{D} \cup \{w_i\}$ with $r$ & Excite $\mathcal{D} \cup \{w_i,w_j\}$ with $r$ & Excite $\mathcal{D} \cup \{w_i\}$ with $r$ &  Excite $\mathcal{D} \cup \{w_j\}$ with $r$ \\ 
 \hline
  Sensor allocation & Measure $\mathcal{D} \cup \{w_i,w_j\}$ & Measure $\mathcal{D}$, measure or indirectly measure $w_j$&  Measure $\mathcal{D} \cup \{w_j\}$  & Measure $\mathcal{D}$ and indirectly measure $w_j$\\ 
 \hline
\end{tabular}
\end{table*}

\begin{theorem} \label{theorem:UnmeasureOutput}
For a model set $\mathcal{M}$ that satisfies Assumptions~\ref{ass:fullrankPhi}, \ref{ass:InPara}, \ref{ass:FeedThrough}, \ref{ass:fixEntry}, suppose that for module $G_{ji}$, its input $w_i$ is measured but output $w_j$ is unmeasured, and let there be a vertex $x_j$ in $\mathcal{X}$, which has out-degree $1$ and a known edge to $w_j$. Then $G_{ji}$ is generically identifiable in $\mathcal{M}$ from $(w_\mathcal{C},r)$ if for some $\bar{\mathcal{X}} \subseteq \mathcal{X}_j \setminus \{x_j\}$ and $\bar{\mathcal{C}}\subseteq \mathcal{C}$ with $w_i \in \bar{\mathcal{C}}$, there exists a $( \mathcal{N}_i^+ \setminus \{w_j\}) \cup \bar{\mathcal{X}} \cup \{w_i\} - \bar{\mathcal{C}}$ disconnecting set $\mathcal{D}$ such that
\begin{enumerate}
\item  $b_{ \{w_j\} \cup \mathcal{D}  \to \bar{\mathcal{C}}} = |\mathcal{D}| + 1;$

\item  $b_{\bar{\mathcal{X}} \to \mathcal{D}}=|\mathcal{D}|$.
\end{enumerate}
\end{theorem}

Note that since $w_i \in \bar{\mathcal{C}}$, the disconnecting set in the above result must contain $w_i$. This result generalizes Theorem~\ref{theorem:partial2} to address the situation where the input is measured but has an indirect excitation source, while the input needs to be excited in Theorem~\ref{theorem:partial2}. The above generalization is achieved by additionally blocking the paths from $w_i$ to $\bar{\mathcal{C}}$ using the disconnecting set, compared to Theorem~\ref{theorem:partial2}. Furthermore, the result in Theorem~\ref{theorem:UnmeasureOutput} can be extended to analyzing global identifiability as in Remark~\ref{remark:globalPath}  and identifiability of a subnetwork, i.e. a subset of out-going modules of $w_i$ in this case, as in Remark~\ref{remark:SIMO}. A path-based formulation of Theorem~\ref{theorem:UnmeasureOutput} can be also obtained analogously as in Corollary~\ref{corollary:UnmeasureInputPath},

\section{Actuator and sensor allocation for identifiability}\label{sec:synthesis}

The results in the previous sections provide analysis results for verifying identifiability for a given configuration of measured and excited signals. In order to extend these results for solving a synthesis problem, i.e. allocating sensors and actuators to achieve identifiability, we extend a reasoning that is originating from \cite{shi2020IFAC} as follows. 

Depending on whether the input or the output of a module is measured, Theorems~\ref{theorem:partial}, \ref{theorem:partial2}, \ref{theorem:UnmeasureInput} and \ref{theorem:UnmeasureOutput} explicitly require the signals in the disconnecting sets to be excited and measured (indirectly) to guarantee single module identifiability. Therefore, the idea is to first compute a disconnecting set and then allocate signals and sensors accordingly. However, the disconnecting sets in the theorems cannot be computed before the  excitation signals $\bar{\mathcal{X}}$ and the measured signals $\bar{\mathcal{C}}$ are specified. Thus, for the disconnecting sets we first provide the necessary conditions that do not rely on the external and measured signals.

\begin{corollary} 
For a model set $\mathcal{M}$ that satisfies Assumptions~\ref{ass:fullrankPhi}, \ref{ass:InPara}, \ref{ass:FeedThrough}, \ref{ass:fixEntry},
\begin{itemize}
\item if it satisfies the conditions in Theorem~\ref{theorem:partial} with disconnecting set $\mathcal{D}_1$, $\mathcal{D}_1$ is also a $\{w_i\} - \mathcal{N}^-_j \setminus \{w_i\}$ disconnecting set;

\item if it satisfies the conditions in Theorem~\ref{theorem:partial2} with disconnecting set $\mathcal{D}_2$, $\mathcal{D}_2$ is also a $\mathcal{N}_i^+ \setminus \{w_j\} - \{w_j\}$ disconnecting set;

\item if it satisfies the conditions in Theorem~\ref{theorem:UnmeasureInput} with disconnecting set $\mathcal{D}_3$, $\mathcal{D}_3$ is also a $\{w_i\} - \mathcal{N}_i^* \cup (\mathcal{N}_j^- \setminus \{w_i\}  )$ disconnecting set;

\item if it satisfies the conditions in Theorem~\ref{theorem:UnmeasureOutput} with disconnecting set $\mathcal{D}_4$, $\mathcal{D}_4$ is also a $( \mathcal{N}_i^+ \setminus \{w_j\}) \cup \{w_i\} - \{w_i,w_j\}$ disconnecting set;

\end{itemize}

\end{corollary}
\begin{proof}
The proof is analogous to the proof of \cite[Corollary~2]{shi2020generic}.
\end{proof}

After computing the above specified disconnecting sets, excitation signals and sensors can be allocated to achieve identifiability of $G_{ji}$. Following the theorems, we illustrate the synthesis approaches in Table~\ref{table:1} for different cases depending on whether the input or the output can be measured. For each situation, we specify how the disconnecting set is constructed and which signals are to be excited or measured. 

Taking Case~2 in Table~\ref{table:1} as an example, a $\mathcal{N}_i^+ \setminus \{w_j\} - \{w_j\}$ disconnecting set $\mathcal{D}$ that satisfies $w_j \notin \mathcal{D}$ is first computed. Then for actuator and sensor allocation, each signal in $\mathcal{D} \cup \{w_i,w_j\}$ is excited by a distinct $r$ signal, and all signals in $\mathcal{D}$ are measured. In addition, the output $w_j$ is indirectly measured, which is also required in Case~4. 

The indirect measurement for $w_j$ in Cases~2 and 4 means that $w_j$ has a path to a measured vertex $w_k$, and more importantly, based on Theorems~\ref{theorem:partial2} and \ref{theorem:UnmeasureOutput}, we choose $w_k$ in the following way: $w_k \in \mathcal{C} \setminus \{w_j\}$ is an internal signal such that the computed $\mathcal{D}$ is also a $\mathcal{N}_i^+ \setminus \{w_i\} - \{w_k\}$ disconnecting set in Case 2 or a $( \mathcal{N}_i^+ \setminus \{w_j\})  \cup \{w_i\} - \{w_k\}$ disconnecting set in Case 4; in addition, there exists a path from $w_j$ to $w_k$ that is vertex disjoint with $\mathcal{D}$.

The synthesis approaches can be justified in the following result.
\begin{theorem} \label{propo:synthesis1}
For a model set $\mathcal{M}$ that satisfies Assumptions~\ref{ass:fullrankPhi}, \ref{ass:InPara}, \ref{ass:FeedThrough}, \ref{ass:fixEntry}, consider each case in Table~\ref{table:1}. If the disconnecting set $\mathcal{D}$ is formulated and the excitation signals and sensors are allocated according to Table~\ref{table:1}, module $G_{ji}$ is generically identifiable in the obtained model set from the measured signals.
\end{theorem}
\begin{proof}
In Case~1, the allocated $r$ signals form a set $\bar{\mathcal{X}}$, and $\mathcal{D}$ forms a set $\bar{\mathcal{C}}$. It is straightforward that the obtained $\bar{\mathcal{X}}$, $\bar{\mathcal{C}}$ and $\mathcal{D}$ together satisfy the conditions in Theorem~\ref{theorem:partial}. In Case~2, the $r$ signals that excite $\mathcal{D}$ form the set $\bar{\mathcal{X}}$, and $\mathcal{D} \cup \{w_k\}$, or $\mathcal{D} \cup \{w_j\}$ if $w_j$ is measured, forms the set $\bar{\mathcal{C}}$ where $w_k$ is the indirect measurement of $w_j$. Then due to the chosen $w_k$ and the direct excitation and measurement of $\mathcal{D}$, the conditions of Theorem~\ref{theorem:partial2} are satisfied. The proofs for the other two cases follows analogously and thus are omitted.
\end{proof}

When computing the disconnecting set in the above approaches, a minimum disconnecting set can be found by the Ford–Fulkerson algorithm in time $\mathcal{O}(|\mathcal{E}||\mathcal{V}|)$ \cite{schrijver2003combinatorial}, where $\mathcal{V}$ and $\mathcal{E}$ are the vertex set and edge set of the dynamic network $\mathcal{G}$, respectively. Also note that in the obtained model set in Cases~1 and 3, module $G_{ji}$ is globally identifiable as all the relevant signals are exited and measured, and thus Theorems~\ref{theorem:partial2} and \ref{theorem:UnmeasureInput} are satisfied with the unique sets of maximum number of vertex disjoint paths. Global identifiability of $G_{ji}$ in the other two cases depends on the chosen indirect measurement $w_k$ for $w_j$.

Furthermore, like the extension made in Corollary~\ref{corollary:partialPathSubnetwork}, the synthesis approaches for Cases~1 and 3 can be extended trivially to deal with a subset of in-coming modules of $w_j$; and the ones for Cases~2 and 4 can also be extended to consider a subset of out-going modules of $w_i$, as noted in Remark~\ref{remark:SIMO}.

\section{Indirect identification methods} \label{sec:Indirec}
The disconnecting-set-based results in Theorems~\ref{theorem:partial}, \ref{theorem:partial2}, \ref{theorem:UnmeasureInput} and \ref{theorem:UnmeasureOutput} also suggest several indirect identification methods through which the module of interest $G_{ji}$ can be estimated in the situation that the identifiability conditions are satisfied through external excitation signals $r$ only, and so no use is made of noise excitation. The identification algorithms can be retrieved from the following result, where all non-zero modules are assumed to be unknown for simplicity.
\begin{proposition}
For $\mathcal{M}$ with all fixed modules being zero,
\begin{itemize}
\item if it satisfies the conditions in Theorem~\ref{theorem:partial} with $w_j \notin \mathcal{D}$ and $\bar{\mathcal{X}}$ having no directed edge to $w_j$, then for $\mathcal{M} $ it holds generically that
\begin{equation}
G_{ji}(q,\theta) = T_{j  \bar{\mathcal{X}}}(q,\theta) \begin{bmatrix}
T_{i \bar{\mathcal{X}}}(q,\theta) \\
T_{\bar{\mathcal{C}} \bar{\mathcal{X}}}(q,\theta)
\end{bmatrix}^\dagger \begin{bmatrix}
1 \\ \mathbf{0}
\end{bmatrix}; \label{eq:indirect1}
\end{equation}
\item if it satisfies the conditions in Theorem~\ref{theorem:partial2} with $w_i \notin \mathcal{D} \cup \bar{\mathcal{C}}$, it holds generically that
\begin{equation}
G_{ji}(q,\theta)= \begin{bmatrix}
1 & \mathbf{0}
\end{bmatrix}\begin{bmatrix}
T_{\bar{\mathcal{C}} j}(q,\theta) &
T_{\bar{\mathcal{C}}\bar{\mathcal{X}}}(q,\theta)
\end{bmatrix}^\dagger T_{\bar{\mathcal{C}} i} ; \label{eq:indirect2}
\end{equation}

\item if it satisfies the conditions in Theorem~\ref{theorem:UnmeasureInput} with $w_j \notin \mathcal{D}$,  $\bar{\mathcal{X}}$ having no directed edge to $w_j$, and $x_i$ being a measured excitation signal, it holds generically that
\begin{equation}
G_{ji}(q,\theta) = T_{j  \bar{\mathcal{X}}}(q,\theta) \begin{bmatrix}
\mathbf{e}_{x_i}\\
T_{\bar{\mathcal{C}} (\bar{\mathcal{X}}\cup \{x_i\})}(q,\theta) 
\end{bmatrix}^\dagger \begin{bmatrix}
1 \\ \mathbf{0}
\end{bmatrix}, \label{eq:indirect3}
\end{equation} 
where $\mathbf{e}_{x_i}$ is a standard basis vector which denotes the mapping from $(\bar{\mathcal{X}}\cup \{x_i\})$ to $x_i$;

\item if it satisfies the conditions in Theorem~\ref{theorem:UnmeasureOutput}, it holds generically that
\begin{equation}
G_{ji}(q,\theta)= \begin{bmatrix}
1 & \mathbf{0}
\end{bmatrix}\begin{bmatrix}
T_{\bar{\mathcal{C}} j}(q,\theta) &
T_{\bar{\mathcal{C}}\bar{\mathcal{X}}}(q,\theta)
\end{bmatrix}^\dagger C_{\bar{\mathcal{C}} i},  \label{eq:indirect4}
\end{equation}
where $C_{\bar{\mathcal{C}} i}$ is the submatrix of $C$ whose rows and columns correspond to $\bar{\mathcal{C}}$ and $w_i$, respectively.
\end{itemize}
\end{proposition}

\begin{proof}
\eqref{eq:indirect1} is obtained from \eqref{eq:proofPar1} where $P = T_{j\bar{\mathcal{X}}}$ under the assumptions that $\bar{\mathcal{X}}$ has no directed edge to $w_j$ and there is no known module, and \eqref{eq:indirect2} is a dual result of \eqref{eq:indirect1}. Similarly, \eqref{eq:indirect3} is derived from  \eqref{eq:UnInputProof2}, and combining \eqref{eq:partialvvproof} and \eqref{eq:partialvvproof2} leads to \eqref{eq:indirect4}.
\end{proof}

The four expressions in the above proposition show opportunities to estimate $G_{ji}$ when the set $\mathcal{X}_j$ consists of only measured $r$ signals. In this case, all the involved mappings, including $T_{\bar{\mathcal{C}} \bar{\mathcal{X}}}$, $T_{\bar{\mathcal{C}} i}$ and $T_{\bar{\mathcal{C}} j}$, are mappings from measured $r$ signals to measured internal signals, and thus, they can be estimated consistently under mild conditions. Consequently, $G_{ji}$ can also be estimated consistently on the basis of the above expressions. Taking \eqref{eq:indirect2} as an example, it shows that an estimate of $G_{ji}$ can be obtained from dividing the mapping from $x_i$ to $\bar{\mathcal{C}}$ by the mapping from $\{x_j \} \cup \bar{\mathcal{X}}$ to $\bar{\mathcal{C}}$, as visualized in Figure~\ref{fig:2ndTheorem}. Note that these methods can also be generalized to identify multiple modules, which is similar to the extensions in Corollary~\ref{corollary:partialPathSubnetwork} and Remark~\ref{remark:SIMO}. 

The methods cover all possible situations for single module identification, depending on if the input or the output is measured. The particularly interesting case is \eqref{eq:indirect2} which leads to a method that can estimate $G_{ji}$ even when both its input and output are unmeasured. The methods also do not require measuring the inputs of the MISO subsystem that contains $G_{ji}$, which is a typical choice in the literature such that the MISO subsystem can be estimated to obtain the module estimate \cite{gevers2018practical}. 

The obtained indirect methods can address significantly more general settings than the existing network identification methods \cite{dankers2015,linder2017identification,gevers2018practical, ramaswamyCDC2019,shi2020generic} which are typically limited to a specific measurement scheme, e.g., both the input and output are measured \cite{dankers2015,gevers2018practical,ramaswamyCDC2019,shi2020generic}; or the input is unmeasured but the output is measured \cite{linder2017identification, gevers2018practical}. The disadvantage of these indirect methods is that they can not exploit excitation through unmeasured disturbance signals, and therefore generally will require a more  expensive experimental setup, requiring a relatively large number of external excitation signals $r$. The presented identifiability conditions point to the use of identification methods that can exploit both types of excitation, as, e.g., \cite{ramaswamyCDC2019}.

\section{Conclusion} \label{sec:conclude}
For single module identifiability analysis, this work introduces the concept of network equivalence and develops a novel network model structure, which has a simple noise model and more importantly, allows us to explore the noise excitation for the identifiability tests.

More importantly, graphical conditions for verifying both global and generic identifiability of a single module are developed in the case of partial measurement and partial excitation. Given the developed model structure, the graphical conditions regard both measured excitation signals and unmeasured noises as excitation sources for identifiability analysis. It is also shown that disconnecting sets provide important information regarding which signals should be excited or measured to achieve identifiability. The above information further leads to synthesis approaches, for excitation allocation and sensor allocation to achieve identifiability, and indirect methods for estimating network modules. The Matlab software for conducting the graphical analysis in Proposition~\ref{proposition:partialPath} and the synthesis in the Case 1 of Table~\ref{table:1} can be found in \cite{ShiCodePartial}.


%

\appendix
\section*{Proof of Theorem~\ref{theorem:equiv}}
We first exploit the spectrum $C \Phi C^\top$ of $M$. Based on the measured signals $w_{\mathcal{C}}$, an immersed network model, which only represents the behavior of the measured signals, can be obtained by eliminating the unmeasured signals (called immersion or Kron reduction) \cite{dankers2015}. To introduce the immersed network, we first define
\begin{align}
\bar{G} &= G_{\mathcal{C}\mathcal{C}}+G_{\mathcal{C}\mathcal{Z}} (I-G_{\mathcal{Z} \mathcal{Z}})^{-1}G_{\mathcal{Z}\mathcal{C}} \label{eq:barG} \\
\bar{H} & =  H_{\mathcal{C}}  +G_{\mathcal{C}\mathcal{Z}} (I-G_{\mathcal{Z} \mathcal{Z}})^{-1}H_{\mathcal{Z}}, \label{eq:barH}
\end{align}
and $\bar{R}$ similarly, where, for example, $G_{\mathcal{C}\mathcal{Z}}$ represents the submatrix of $G$ that has its rows and columns corresponding to the signals in $\mathcal{C}$ and $\mathcal{Z}$, respectively. Then the immersed network model after the elimination of unmeasured internal signals has the following form:
\begin{equation}
w_{\mathcal{C}} = \bar{G} w_{\mathcal{C}} + \bar{R} r(t) + \bar{H} e(t). \label{eq:obsermodel}
\end{equation}
Note that the above model may have non-zero diagonal entries in $\bar{G}$, and $(I-\bar{G})$ has a proper inverse because of Assumption~\ref{ass1}(c) and consequently $(I-\bar{G}^\infty)$ having full rank. This model further leads to an external-to-internal mapping:
\begin{equation}
w_{\mathcal{C}} = (I-\bar{G})^{-1} \bar{R}r(t) + (I-\bar{G})^{-1} \bar{H}e(t). \label{eq:oberMap}
\end{equation}
Based on \eqref{eq:origiNetMap} and \eqref{eq:oberMap}, it can be found that 
\begin{equation}
 C\Phi C^T = (I-\bar{G})^{-1} \bar{H} \Lambda \bar{H}^* (I-\bar{G})^{-*}, \label{eq:proofStart}
\end{equation}
where it holds that
\begin{equation}
C(I-G)^{-1}C^T = (I-\bar{G})^{-1}. \label{eq:proof1a}
\end{equation}
In addition, $\bar{H}\Lambda \bar{H}^\star$ can be re-factorized into $\tilde{H} \tilde{\Lambda} \tilde{H}^\star$ \cite{Gevers2019Singular}, where $(\tilde{H}, \tilde{\Lambda})$ satisfies the properties of this theorem. Note that $(\tilde{H}, \tilde{\Lambda})$ is unique and $\tilde{\Lambda}$ has full rank if Assumption~\ref{ass:fullrankPhi} is satisfied. This together with \eqref{eq:proofStart} and \eqref{eq:proof1a} leads to
\begin{equation*}
C\Phi C^T = C(I-G)^{-1} \begin{bmatrix}
\tilde{H} \\
0
\end{bmatrix} \tilde{\Lambda} \begin{bmatrix}
\tilde{H}^\star &0
\end{bmatrix} (I-G)^{-*} C^T.
\end{equation*}
The above equation implies that the external-to-output mapping of \eqref{eq:netNh}, i.e.
$$
w_{\mathcal{C}} = C(I-G)^{-1}R r +C(I-G)^{-1}\begin{bmatrix}
\tilde{H} \\
0
\end{bmatrix} \tilde{e},
$$
leads to the same object $(CTR, C \Phi C^\top)$ as $M$, which concludes that $\tilde{M} \sim M$. 

\section*{Proof of Proposition~\ref{pro:reforIden}}
For any $\theta$, it holds that
$$ C\Phi_{v}(\theta)C^T =  C(I-G)^{-1} H_p \tilde{\Lambda} H_p^\star (I-G)^{-*} C^T, 
$$
where $H_p \triangleq \begin{bmatrix}
\tilde{H} \\
0
\end{bmatrix} $, $\tilde{H}$ is monic and $C=[I \quad 0]$. Then the proposition is proved by showing that unique $C(I-G)^{-1} H_p$ and $\tilde{\Lambda}$ can be found given $CTR$ and $C \Phi C^\top$ under Assumption~\ref{ass:FeedThrough}, and then the two implications are trivially equivalent.

If $G$ is strictly proper by Assumption~\ref{ass:FeedThrough}(a), $C(I-G)^{-1} H_p$ is also monic. As $C\Phi_{v}(\theta)C^T$ admits a unique factorization as 
$$C\Phi_{v}(\theta)C^T = L \Lambda_p L^\star,$$ 
where $L$ is monic \cite{Gevers1981SpectrumFactor}, it holds that $C(I-G)^{-1} H_p = L$ and $\tilde{\Lambda} = \Lambda_p$. Thus, the uniqueness of $\Lambda$ and $C(I-G)^{-1} H_p$ is guaranteed given $C\Phi_{v}(\tilde{\theta})C^T$, which concludes the proof under Assumption~\ref{ass:FeedThrough}(a).

If Assumption~\ref{ass:FeedThrough}(b) holds, there exists a permutation matrix $P$ such that 
$$C\Phi_{v}(\theta)C^T = (CP^\top) F  (P H_p) \tilde{\Lambda} (H_p^\star P^\top) F^* (P C^\top), $$ 
where $F\triangleq [I-P G^\infty (\theta) P^\top]^{-1}$ and $F$ is lower unitriangular. As $CP^\top$ contains the first $|\mathcal{C}|$ rows of $P^\top$, there exists another permutation matrix $\bar{P}$ such that $\bar{P}CP^\top$ is in row echelon form. Note that pre-multiplying a square matrix by $\bar{P}CP^\top$ extracts a subset of rows in the matrix without re-ordering them. Based on the above facts, consider $\bar{P} C\Phi_{v}(\theta)C^T \bar{P}^\top$ that equals
$$ ( \bar{P} C P^\top) F  (P H_p  \bar{P}^\top) \bar{P} \tilde{\Lambda} \bar{P}^\top ( \bar{P} H_p^\star P^\top) F^* (P C^\top \bar{P}^\top),$$
where $( \bar{P} C P^\top) F  (P H_p  \bar{P}^\top)$ is lower unitriangular because the pre- and post-multiplication of $F$ leads to a submatrix of $F$ with rows and columns corresponding to the same indexes. As $\bar{P} C\Phi_{v}(\theta)C^T \bar{P}^\top$ admits a unique $LDL^\star$ decomposition as 
$$
\bar{P} C\Phi_{v}(\theta)C^T \bar{P}^\top = \bar{L} \Lambda_p \bar{L}^\star,
$$
where $\bar{L}$ is lower unitriangular, $( \bar{P} C P^\top) F  (P H_p  \bar{P}^\top)$ and $\bar{P} \tilde{\Lambda} \bar{P}^\top$ can be uniquely determined as 
$$
( \bar{P} C P^\top) F  (P H_p  \bar{P}^\top) = \bar{L}, \quad  \bar{P} \tilde{\Lambda} \bar{P}^\top= \Lambda_p.
$$
Therefore, $C(I-G)^{-1}H_p$ and $\tilde{\Lambda}$ can also be uniquely found given the spectrum matrix, which proves the proposition under Assumption~\ref{ass:FeedThrough}(b).

\section*{Proof of Lemma~\ref{lemma:K} }
The existence of $K$ is proven in Theorem~5 of \cite{shi2020generic}, and based on this theorem, $K$ can take the following form:
$$
K = \begin{bmatrix}
[(I- G_{\mathcal{P}\mathcal{P}})^{-1}]_{\mathcal{W}_1 \star}G_{\mathcal{P}\mathcal{D}_1} & [(I- G_{\mathcal{P}\mathcal{P}})^{-1}]_{\mathcal{W}_1 \star}X_{\mathcal{P} \mathcal{D}_2} \\
C & 0
\end{bmatrix},
$$
for some $\mathcal{P} \subseteq \mathcal{W}$, $\mathcal{D}_1 = \mathcal{W} \cap \mathcal{D}$ and $\mathcal{D}_2 = \mathcal{X} \cap \mathcal{D}$; $C$ is a selection matrix that extracts the rows corresponding to $\mathcal{W}_2 $ from a matrix whose rows correspond to $\mathcal{D}_1$, where $\mathcal{W}_2 = \bar{\mathcal{W}} \cap \mathcal{D}$, and $\mathcal{W}_1 = \bar{\mathcal{W}} \setminus \mathcal{W}_2$.

Note that the $K$ matrix equals the external-to-internal mapping from $\mathcal{D}$ to $\bar{\mathcal{W}}$ in a subgraph of $\mathcal{G}$, where the vertices in $\mathcal{D}$ are the external signals with all in-coming edges of $\mathcal{D}$ removed, and the signals in $\mathcal{P}$ are internal signals. This characterization of $K$ is clearly seen from its formulation: $(I-G_{\mathcal{P}\mathcal{P}})^{-1}\begin{bmatrix}
G_{\mathcal{P}\mathcal{D}_1} & X_{\mathcal{P}\mathcal{D}_2}
\end{bmatrix}$
has the same structure as $T_{\mathcal{W}\mathcal{X}}$ in \eqref{eq:TwxDef}; and the block row $[C \text{ }0]$ in $K$ represents the mapping from $\mathcal{D}$ to $\mathcal{W}_2$, as $C$ is the mapping from $\mathcal{D}_1$ to $\mathcal{W}_2 \subseteq \mathcal{D}_1$, and the zero entries indicate that there is no path between any two distinct vertices in $\mathcal{D}$ when vertices in $\mathcal{D}$ are externals signals. Thus, the full column rank of $K$ can be evaluated based on Lemma~\ref{lemma:RankPath} for this subgraph.

\section*{Proof of Theorem~\ref{theorem:partial}  }
The proof is to show that a unique $G_{j i}$ can be found given $T_{\bar{\mathcal{C}} \bar{\mathcal{X}}} $, $T_{i \bar{\mathcal{X}}}$ and $T_{j \bar{\mathcal{X}}}$. Note that by condition~\eqref{partial1:con1}, $w_i \notin \mathcal{D}$ holds. Let set $\mathcal{N}$ contain the remaining in-coming internal signals of $w_j$ which are not in $\mathcal{N}^-_j$, i.e. $\mathcal{N}$ contains the ones that are measured and have known directed edges to $w_j$. When $w_j \notin \mathcal{D}$, recall the $j$th row of $(I-G)T_{\mathcal{W}\mathcal{X}} = X$, and after permutation we have
\begin{equation}
\begin{bmatrix}
- G_{ji} & -G_{j \mathcal{N}^-_j \setminus \{w_i \}} K_1 & 1 & 0 
\end{bmatrix} \begin{bmatrix}
T_{i \bar{\mathcal{X}}} \\
T_{\mathcal{D} \bar{\mathcal{X}}} \\
T_{j \bar{\mathcal{X}}}\\
\star
\end{bmatrix} = \bar{X}+G_{j\mathcal{N}}T_{\mathcal{N}\bar{\mathcal{X}}} , \label{eq:partiproof1} 
\end{equation}
where $\bar{X}$ is a row vector with its columns corresponding to the signals in $\bar{\mathcal{X}}$ and thus is known; $G_{j\mathcal{N}}T_{\mathcal{N}\bar{\mathcal{X}}}$ is also given as the modules in $G_{j\mathcal{N}}$ are known and $\mathcal{N}$ is measured; $K_1$ satisfies $K_1 T_{\mathcal{D} \bar{\mathcal{X}}} = T_{\mathcal{N}^-_j \setminus \{w_i \} \bar{\mathcal{X}}}$ based on Lemma~\ref{lemma:K}. Furthermore, there exists $K_2$ such that $K_2 T_{\mathcal{D} \bar{\mathcal{X}}} = T_{\bar{\mathcal{C}} \bar{\mathcal{X}}}$, and $K_2$ has full column rank generically by condition~\eqref{partial1:con2} and Lemma~\ref{lemma:K}. This generically leads to 
\begin{equation}
T_{\mathcal{D} \bar{\mathcal{X}}} = K_2^\dagger T_{\bar{\mathcal{C}} \bar{\mathcal{X}}}, \label{eq:proofPar2}
\end{equation}
where $()^\dagger$ denotes the matrix's left inverse. Then combining the above equation and \eqref{eq:partiproof1} leads to
\begin{equation}
\begin{bmatrix}
- G_{j i} &-G_{j\mathcal{N}^-_j \setminus \{w_i \}} K_1 K_2^\dagger 
\end{bmatrix} \begin{bmatrix}
T_{i \bar{\mathcal{X}}} \\
T_{\bar{\mathcal{C}} \bar{\mathcal{X}}} 
\end{bmatrix} = P, \label{eq:proofPar1}
\end{equation}
where $P=\bar{X} - T_{j \bar{\mathcal{X}}}+G_{j\mathcal{N}}T_{\mathcal{N}\bar{\mathcal{X}}}$, and the above equation holds generically. In addition, due to conditions~\eqref{partial1:con1} and \eqref{partial1:con2}, it holds that 
\begin{equation*}
b_{\bar{\mathcal{X}}  \to \{w_i\}\cup \bar{\mathcal{C}}}=1+ b_{\bar{\mathcal{X}}  \to \bar{\mathcal{C}}},
\end{equation*}
and thus generically
\begin{equation}
rank(T_{(\{w_i\}\cup \bar{\mathcal{C}}) \bar{\mathcal{X}}})=1+ rank(T_{\bar{\mathcal{C}}\bar{\mathcal{X}}}), \label{eq:proofPar3}
\end{equation}
which implies that \eqref{eq:proofPar1} has a unique solution for $G_{ji}$ generically based on \cite[Lemma~2]{shi2020generic} and thus generic identifiability of $G_{ji}$.

When $w_j \in \mathcal{D}$, the $j$th row of $(I-G)T_{\mathcal{W}\mathcal{X}}=X$ can be written as follows after permutation:
\begin{equation*}
\begin{bmatrix}
- G_{ji} & -G_{j (\{w_j\} \cup \mathcal{N}^-_j \setminus \{w_i\})}\bar{K}_1 & 0 
\end{bmatrix} \begin{bmatrix}
T_{i \bar{\mathcal{X}}} \\
T_{\mathcal{D} \bar{\mathcal{X}}} \\
\star
\end{bmatrix} =  \bar{X}+G_{j\mathcal{N}}T_{\mathcal{N}\bar{\mathcal{X}}},  
\end{equation*}
where $\bar{K}_1 T_{\mathcal{D} \bar{\mathcal{X}}} = T_{( \{w_j\} \cup \mathcal{N}^-_j \setminus \{w_i \}  ) \bar{\mathcal{X}}}$ for some $\bar{K_1}$.
Note that for the above equation in the special case where $\mathcal{N} = \emptyset$, $G_{j\mathcal{N}}T_{\mathcal{N}\bar{\mathcal{X}}}$ disappears and $\bar{X}$ becomes non-zero, because $\bar{\mathcal{X}}$ must have a directed edge to $w_j \in \mathcal{D}$; otherwise, there exists a path from $\bar{\mathcal{X}}$ to $\mathcal{N}^-_j \setminus \{w_i\}$ which does not intersect with $\mathcal{D}$ based on condition~(1) and thus contradicts that $\mathcal{D}$ is a disconnecting set.
Finally, combining the above equation and \eqref{eq:proofPar2} leads to unique $G_{ji}$ generically given $CT_{\mathcal{W} \bar{\mathcal{X}}}$ and the first two conditions. 

\section*{Proof of Proposition~\ref{proposition:partialPath}}
We prove this result by showing that the conditions are equivalent to the conditions of Theorem~\ref{theorem:partial}. Firstly, based on \cite[Lemma~8]{shi2020generic}, condition~\eqref{partial1Path:con1} is equivalent to condition~\eqref{partial1:con1} of Theorem~\ref{theorem:partial}, and both conditions are satisfied by choosing $\mathcal{D}$ to be a minimum $\bar{\mathcal{X}} - (\mathcal{N}^-_j\setminus \{w_i\} ) \cup \bar{\mathcal{C}}$ disconnecting set. With this choice, $b_{\bar{\mathcal{X}} \to (\mathcal{N}^-_j \setminus \{w_i\} ) \cup \bar{\mathcal{C}}} = |\mathcal{D}|$ by the Menger's theorem \cite{shi2020generic}, and $\mathcal{D}$ is also a $\bar{\mathcal{X}} - \bar{\mathcal{C}}$ disconnecting set. Then if condition~\eqref{partial1Path:con2} is satisfied, it holds that
$$
|\mathcal{D}| = b_{\bar{\mathcal{X}} \to \bar{\mathcal{C}}},
$$
which implies that $\mathcal{D}$ is a minimum $\bar{\mathcal{X}}-\bar{\mathcal{C}}$ disconnecting set by the Menger's theorem, and thus $b_{\mathcal{D} \to \bar{\mathcal{C}}} = |\mathcal{D}|$, i.e. condition~\eqref{partial1Path:con2} in Theorem~\ref{theorem:partial}. On the other hand, if condition~\eqref{partial1Path:con2} in Theorem~\ref{theorem:partial} holds, there are maximally $|\mathcal{D}|$ vertex disjoint paths from $\bar{\mathcal{X}}$ via $\mathcal{D}$ to $\bar{\mathcal{C}}$, collected into set $\mathcal{P}$. As $\mathcal{D}$ is a minimum $\bar{\mathcal{X}} - (\mathcal{N}^-_j\setminus \{w_i\} ) \cup \bar{\mathcal{C}}$ disconnecting set, all the paths from $\bar{\mathcal{X}}$ to $(\mathcal{N}^-_j\setminus \{w_i\} )$ should intersect with the paths in $\mathcal{P}$, which implies condition~(2) of this result and thus concludes the proof.

\section*{Proof of Theorem~\ref{theorem:UnmeasureInput}}
When $w_j \notin \mathcal{D}$, as the conditions imply the first two conditions in Theorem~\ref{theorem:partial} with $\bar{\mathcal{X}}\cup \{x_i\}$, we can use part of the proof for Theorem~\ref{theorem:partial}, while the differences start from \eqref{eq:proofPar1} by replacing $\bar{\mathcal{X}}$ in \eqref{eq:proofPar1} with $\tilde{\mathcal{X}}$, where $\tilde{\mathcal{X}} = \bar{\mathcal{X}}\cup \{x_i\}$. In addition, it can be found that $\mathcal{D} \cup \{x_i\}$ is a $\tilde{\mathcal{X}} -\{w_i\}$ disconnecting set, and thus for some proper $K_i$, it holds
\begin{equation}
T_{i\tilde{\mathcal{X}}} = K_i \begin{bmatrix}
\textbf{e}_{x_i} \\
T_{\mathcal{D} \tilde{\mathcal{X}}} 
\end{bmatrix}, \label{eq:UnInputProof1}
\end{equation}
where $\textbf{e}_{x_i}$ is a row vector contains one entry as $1$ and zeros elsewhere, and it denotes the mapping from $\tilde{\mathcal{X}}$ to $x_i$ since $x_i \in \tilde{\mathcal{X}}$. Moreover, following the proof of Lemma~\ref{lemma:K}, $K_i$ is the external-to-internal mapping of a subnetwork, where $\mathcal{D} \cup \{x_i\}$ are external signals, and all in-coming edges of $\mathcal{D}$ are removed. As $\mathcal{D} \cup \{x_i\}$ intersects with all the paths from $w_i$ to $\mathcal{N}_i^*$ in $\mathcal{G}$, the subnetwork does not contain any loop around $w_i$. This indicates that $K_i$ has a special structure as 
$$
K_i = \begin{bmatrix}
\bar{H} & \bar{K}_i
\end{bmatrix},
$$
where the $\bar{H}$ is the \textit{known} module from $x_i$ to $w_i$ in the subnetwork, and note that $\bar{H} =1$ if $x_i$ is a signal in $r$. Combining the above equation, \eqref{eq:UnInputProof1}, \eqref{eq:proofPar2} and \eqref{eq:proofPar1} leads to 
\begin{equation}
\begin{bmatrix}
- G_{j i} \bar{H} &(- G_{j i} \bar{K}_i-G_{j \mathcal{N}^-_j \setminus \{w_i \}} K_1 )K_2^\dagger 
\end{bmatrix} \begin{bmatrix}
 \textbf{e}_{x_i} \\
T_{\bar{\mathcal{C}} \tilde{\mathcal{X}}} 
\end{bmatrix} = P, \label{eq:UnInputProof2}
\end{equation}
where $\begin{bmatrix}
\textbf{e}_{x_i} \\
T_{\bar{\mathcal{C}} \tilde{\mathcal{X}}} 
\end{bmatrix} $ denotes the mapping from $\tilde{\mathcal{X}}$ to $\bar{\mathcal{C}}\cup \{x_i\}$. Following a similar reasoning as in the proof of Lemma~\ref{lemma:K}, it can be obtained from the two graphical conditions that generically $G_{ji}\bar{H}$ and consequently $G_{ji}$ are obtained uniquely since $\bar{H}$ is known. The proof for the case where $w_j \in \mathcal{D}$ can be shown analogously.

\section*{Proof of Theorem~\ref{theorem:UnmeasureOutput}}
Firstly, it holds that $w_i \in \mathcal{D}$ since $\bar{\mathcal{C}}$ contains $w_i$, and let $\mathcal{N}$ denote the out-neighbors of $w_i$ that are not in $\mathcal{N}_i^+$. Considering the column of $C= T_{\mathcal{C}\mathcal{W}}(I-G)$ corresponding to $w_i$ and since $\mathcal{D}$ is a $( \mathcal{N}_i^+ \setminus \{w_j\}) \cup \{w_i\} - \bar{\mathcal{C}}$ disconnecting set, we have 
\begin{equation}
\begin{bmatrix}
T_{\bar{\mathcal{C}} j} &
T_{\bar{\mathcal{C}}\mathcal{D}} \bar{K}_1 
\end{bmatrix} \begin{bmatrix}
- G_{ji} \\
\begin{bmatrix}
-G_{\mathcal{N}_i^+ \setminus \{w_j \} i} \\ 1 \end{bmatrix}
\end{bmatrix} = P,  \label{eq:partialvvproof}
\end{equation}
where $T_{\bar{\mathcal{C}}\mathcal{D}} \bar{K}_1  = T_{\bar{\mathcal{C}} ( \mathcal{N}_i^+ \setminus \{w_j\}) \cup \{w_i\}}$ for some $\bar{K}_1$ based on Lemma~\ref{lemma:Kdual}; and $P=C_{\bar{\mathcal{C}} i}+T_{\bar{\mathcal{C}}\mathcal{N}}G_{\mathcal{N}i}$ where $P$ is known and $C_{\bar{\mathcal{C}} i}$ is now non-zero because $w_i \in \bar{\mathcal{C}}$. Note that $T_{\bar{\mathcal{C}}j}$ is given by $CT_{\mathcal{W}\mathcal{X}}$ due to the existence of $x_j$. In addition, it holds that for some $\bar{K}_2$
\begin{equation}
T_{\bar{\mathcal{C}}\mathcal{D}} = T_{\bar{\mathcal{C}}\bar{\mathcal{X}}}\bar{K}_2^\dagger, \label{eq:partialvvproof2}
\end{equation}
based on Lemma~\ref{lemma:Kdual}. Combining the above equation and \eqref{eq:partialvvproof}, and following the analysis of Theorem~\ref{theorem:partial} analogously conclude that unique $G_{ji}$ is guaranteed generically.

\ifCLASSOPTIONcaptionsoff
  \newpage
\fi

\bibliographystyle{IEEEtran}
\bibliography{ShiLibrary}

\end{document}